\documentclass[aps,pra,showpacs,twoside,twocolumn,10pt]{revtex4-2}
\usepackage[colorlinks=true, citecolor=red, urlcolor=blue ]{hyperref}
\usepackage{times,epsfig,amssymb,amsfonts,amsmath,bm,subfigure,mathtools,amsthm,braket,soul,enumitem,color,physics,graphics,graphicx}
\usepackage[normalem]{ulem}
\usepackage{comment}
\usepackage{xcolor}
\usepackage{epstopdf}
\usepackage{float}
\newtheorem{theorem}{Theorem}
\newtheorem{corollary}{Corollary}

\newtheorem{proposition}{Proposition}

\newtheorem{property}{Property}

\newcommand{\sm}[1]{{\color{magenta}#1}}
\usepackage{xcolor}

\usepackage{optidef}
\usepackage{comment}

\begin{document}

\title{Imaginarity-generating power of unitaries: A resource-theoretic approach}

\author{Akhil Kumar Awasthi$^{1,2}$, Mrinmoy Samanta$^{1,2}$, Sudipta Mondal$^{1,2}$, Ayan Patra$^{1,2}$, Aditi Sen(De)$^{1,2}$}

\affiliation{\(^1\)Harish-Chandra Research Institute, Chhatnag Road, Jhunsi, Allahabad - 211019, India\\
\(^2\) Homi Bhabha National Institute, Training School Complex, Anushakti Nagar, Mumbai 400 094, India}

\begin{abstract}


Imaginarity, stemming from the complex structure of quantum mechanics, has recently emerged as a fundamental resource, yet its dynamical generation remains largely unexplored. In this work, we introduce the notion of imaginarity-generating power (IGP) of unitary dynamics, which quantifies the ability of unitary operations to produce imaginarity from initially real quantum states. To quantify imaginarity, we employ a measure based on the Hilbert--Schmidt norm, which we show to be monotone under real unital operations. Within the framework of dynamical resource theories, we derive an exact expression for the purity-constrained IGP in arbitrary dimensions and show that, for pure real input states, it depends solely on intrinsic and experimentally accessible properties of the unitary. We further analyze its average behavior over ensembles of states with varying purity under both uniform and Hilbert--Schmidt distributions. We prove that it satisfies the essential properties of a valid resource monotone within the dynamical resource theory of imaginarity. We also characterize the  unitaries that maximize the IGP and determine the corresponding bounds. Moreover, for Haar-random unitaries, we show that the IGP concentrates near its maximal value in high dimensions with small fluctuations, indicating that typical high-dimensional quantum dynamics are highly effective at generating imaginarity.

\end{abstract}

\maketitle

\section{Introduction}
\label{sec:intro}

Quantum resource theories (QRTs) provide a {cohesive} and systematic framework to study the manipulation of quantum states and processes under restricted sets of physically allowed operations~\cite{chitambar2019}. By incorporating operational constraints arising from experimental limitations, symmetry considerations, or fundamental physical principles, QRTs identify certain quantum features as resources and enable their quantitative and operational characterization. 
Distinct choices of allowed operations give rise to different resource theories, such as entanglement under local operations and classical communication (LOCC) \cite{Bennett96, Bennett1996, Vedral97, Vedral98, vidal99, Horo09, Vidal2000}, coherence under incoherent operations \cite{Baumgratz2014, Streltsov2017_re, Winter2016prl, Streltsov2015, Streltsov2018}, athermality under thermal operations \cite{Brando2013, Janzing2000, Masanes2017, Brando2015, Gour2018, Ng2018}, and non-stabilizerness under stabilizer operations \cite{Veitch2014, Howard2017, knill2004, Liu2022, Leone2024, Bravyi2005}. Since quantum systems in a realistic situation inevitably interact with the environment, leading to decoherence and degradation of useful resources, the resource-theoretic approach provides a natural and effective strategy for studying quantum systems. Moreover, connections among different resource theories often reveal deeper insight into the structure of nonclassicality.  

Building on this broader perspective, it is natural to ask what fundamental features of quantum theory itself can be treated as resources. Quantum mechanics is formulated over complex Hilbert spaces, where complex phases play a central role in interference and other nonclassical phenomena. This observation has led to the introduction of imaginarity as a resource within a rigorous resource-theoretic setting~\cite{Hickey2018, chitambar2019}. Despite this formal development, its operational significance is not yet fully understood. This question is closely tied to the long-standing debate on whether complex numbers are truly essential to quantum theory, as several works have explored formulations based purely on real Hilbert spaces~\cite{Wootters2012,Hardy2012,Aleksandrova2013,wootters2013optimal,Moretti2017}. However, recent results demonstrate that a formulation of quantum theory based solely on real Hilbert spaces is insufficient to reproduce all physical phenomena~\cite{Renou2021, Li2022, Chen2022_ruling, wu22, Bednorz2022, yao2024}, thereby highlighting the essential role of imaginarity. Consequently, increasing attention has been devoted to understand imaginarity as a resource, with studies revealing its relevance in a variety of quantum information processing tasks~\cite{wu2021, zhu2021, Miyazaki2022, Sajjan2023, Budiyono2023, wei2024_non, Haug2025}.

Motivated by these developments, it becomes crucial to understand not only what imaginarity of a quantum system is, but also how it is generated and transformed under quantum dynamics. The ability of quantum operations to create resources has long been studied through concepts such as entangling power \cite{Zanardi2000, Linowski2020, Eisert21, Qiu25, mondal25, mrinmoy26, samanta2025, cho26} and coherence-generating power \cite{Zanardi2017, Zanardi17, Styliaris18}, which capture how dynamics create useful quantum resources from initially resourceless (free) states. {In a similar spirit, understanding the generation of imaginarity can also be relevant for quantum computation and circuit design, where unitary gates are the elementary building blocks of quantum algorithms \cite{Nielsen2012}. Since complex phases introduced by gates govern interference patterns and computational pathways, quantifying the ability of gates or circuits to generate imaginarity can provide new insight into circuit complexity and may also help to identify classes of gates that are especially powerful for preparing states or dynamics inaccessible within real-valued quantum circuits.}

In this spirit,  we introduce the notion of imaginarity-generating power (IGP) of unitary dynamics, which quantifies the average power of a unitary operator to generate imaginarity from initially real quantum states. To analyze this systematically, we employ the framework of dynamical resource theories (DRTs) where resources are attributed to quantum operations instead of quantum states. 
To quantify IGP, we use a measure of imaginarity based on the Hilbert--Schmidt norm, which we show to be a valid monotone under real unital operations and which also admits a tractable analytical treatment. Within this framework, we derive an exact analytical expression for the purity-constrained imaginarity-generating power of arbitrary unitary operators in arbitrary dimensions. We further demonstrate that, when restricting imaginarity generation from pure real states, the IGP admits a particularly simple form determined solely by the properties of the unitary. We also investigate the average behavior of the IGP over ensembles of input states with varying purity, considering both uniform and Hilbert-Schmidt distributions, which are accessible across different experimental platforms. Moreover, we demonstrate that the IGP constitutes a proper resource monotone for unitary dynamics within the dynamical resource theory of imaginarity, obeying the fundamental requirements of positivity, faithfulness, invariance, and monotonicity under free superoperations. We additionally characterize the family of unitary operators that attain the maximal IGP and derive the corresponding upper bounds. By examining the mean and variance of the IGP over Haar-random unitaries, we prove that in the large-dimension limit, the generic state possesses a near-maximal imaginarity with only negligible fluctuations, thereby revealing the typical tendency of high-dimensional quantum dynamics to generate strong imaginarity.

The remainder of this paper is organized as follows. In Sec.~\ref{sec:RTI}, we develop the resource theory of imaginarity in both static and dynamical settings. In Sec.~\ref{sec:igpudvr}, we propose the notion of imaginarity-generating power of unitary operator and establish that it constitutes a valid resource of imaginarity in the dynamical framework. One of the main results is presented in Sec.~\ref{sec:ipqu}, where we derive the purity-constrained IGP for unitary operators in arbitrary dimensions and analyze its average behavior over the purity. In Sec.~\ref{sec:pipd}, we report a protocol for the experimental estimation of the imaginarity-generating power. The properties of the IGP are discussed in Sec.~\ref{subsec:resource_imag_power}. Furthermore, in Sec.~\ref{sec:typical_unitaries}, we investigate the behavior of typical unitaries within the framework of the DRT of imaginarity, where we compute the purity-constrained average IGP for Haar-random unitaries, examine its deviation from the maximal value, and evaluate its variance. Finally, we conclude in Sec.~\ref{sec:con}.

\section{Resource theory of imaginarity}
\label{sec:RTI}

{Understanding how imaginarity in quantum systems behaves under different physical processes {has recently emerged as} an important direction of exploration, particularly in view of its emerging significance in the foundational aspects of quantum mechanics~\cite{Renou2021, Li2022, Chen2022_ruling, wu22, Bednorz2022, yao2024} as well as {its role} in various quantum information processing tasks~\cite{Elliott2025,Miyazaki2022,wu2021,Wu2024}. A systematic study of imaginarity within the framework of resource theory is, therefore, necessary to understand its usefulness, as well as how it can be generated, manipulated, and preserved under physically relevant operations. In resource theory, there are two main frameworks: static resource theory (SRT) and dynamical resource theory (DRT). The static approach focuses on identifying and quantifying resourceful states, whereas the dynamical {one} studies how quantum operations {or processes} can create or preserve this resource. In the following, we briefly discuss these two perspectives from the viewpoint of imaginarity.}

\subsection{Static resource theory of imaginarity}
\label{sec:static}

{Here, we introduce the definition, fundamental properties, and quantification of imaginarity in quantum states from the perspective of static resource theory (SRT). Throughout our analysis, we consider an arbitrary finite dimensional Hilbert space $\mathbb{C}^d$ of dimension $d$. We denote the set of density operators (quantum states) acting on $\mathbb{C}^d$ by $\mathcal{D}_d=\{\rho:\rho\geq0~\text{and}~\text{tr}(\rho)=1\}$ and the set of all unitary operators $\mathcal{U}^d$ acting on $\mathbb{C}^d$. Moreover, let $\Lambda^{d_1 \to d_2}$ be the set of all completely positive trace-preserving (CPTP) maps, equivalently known as quantum channels, that maps density matrices from $\mathcal{D}_{d_1}$ to $\mathcal{D}_{d_2}$, i.e., \(\Lambda^{d_1 \to d_2} = \{\Phi : \rho_1 \in \mathcal{D}_{d_1} \xrightarrow{\Phi} \rho_2 \in \mathcal{D}_{d_2}\}\). 
{A static resource theory is specified by the following three essential components:}
(i) {\it Free operations} ($\mathcal{F_O}\subseteq\Lambda^{d_1\to d_2}$) -- a restricted class of physically allowed quantum operations, tailored to some specific scenario; 

(ii) {\it Free states} ($\mathcal{F_S}\subseteq\mathcal{D}_d$) -- a set of certain physically realizable states  that can be prepared using the set of free (auxiliary) operations, possibly with an aid of another free state; and 

(iii) {\it Resource monotone} (\(\mathcal{X}\)) -- a function that assigns a non-negative real value to each density matrix, i.e., \(\mathcal{X}:\mathcal{D}_d\to\mathbb{R}_{\geq0}\), with the defining property that it does not increase under free operations, namely, i.e., $\mathcal{X}(\rho)\geq\mathcal{X}(\Phi(\rho))~\forall~\Phi\in\mathcal{F_O}$. In essence, it serves as a quantitative measure of the resource present in a given state within the underlying resource-theoretic framework.}

{In the resource theory of imaginarity, the set of free states comprises all density matrices with real entries, i.e., $\{\rho \in \mathcal{D}_d:\rho_{ij}\in\mathbb{R}~\forall~i,j\}$, which we henceforth refer to as \textit{real states}. Accordingly, the free operations are those that preserve this set of real states. In particular, the class of free operations consists of all CPTP maps that admit Kraus decompositions with real Kraus operators, i.e., $\{\Phi:\Phi(\rho)=\sum_\mu K_\mu\rho K_\mu^\dagger,~\text{with}~[K_\mu]_{ij}\in\mathbb{R}~\forall~\mu,i,j\}$. Such maps are called \textit{real CPTP operations} \cite{BethRuskai2002,kraus1983,Hickey2018}. It is therefore evident that the SRT of imaginarity is a basis-dependent resource theory, since the characterization of states and operations being real is determined by the choice of basis. Accordingly, we fix the reference basis to be the computational basis $\mathbf{B} = \{\ket{i}\}_{i=1}^d$, unless stated otherwise. Having defined real states and real CPTP operations, the next step is the quantification of imaginarity -- \textit{imaginarity monotone}. Several imaginarity monotones have been proposed in the literature~\cite{Du2025_pra,Hickey2018,wu2021}, depending on the underlying operational or mathematical viewpoint. Notable examples include distance-based measures, such as the relative entropy of imaginarity \cite{Wu_2025}; norm-based measures, e.g., the $l_1$-norm of imaginarity~\cite{Guo2025,Chen_2023}; and operationally motivated monotones, including the robustness of imaginarity~\cite{Zhou_2025}, as well as distillable imaginarity  and imaginarity cost~\cite{wu2021,Wu2024}. 

In this work, we employ the $l_2$-norm to quantify imaginarity, it is worth noting that, while the $l_1$-norm of imaginarity is monotone under the entire class of real operations, the $l_2$-norm of imaginarity serves as a valid monotone only under \textit{real unital operations}, defined for any density matrix $\rho\in\mathcal{D}_d$, as 
\begin{eqnarray}
\mathcal{I}_{\mathbf{B}}(\rho) &=& \bigg\| \frac{\rho - \rho^{*}}{2} \bigg\|_{2}^{2}=\frac{1}{4}||\rho-\rho^T||_2^2,
\label{eq:imaginarity_measure}
\end{eqnarray} 
where $\rho^*$ and $\rho^T$ denote the complex conjugate and the transpose of $\rho$, respectively, and  $||A||_2 = \sqrt{\mathrm{tr}(A^\dagger A)}$ is the $l_2$-(Hilbert-Schmidt) norm. The measure has several desirable properties -- (i) $\mathcal{I}_{\mathbf{B}}(\rho) \geq 0$, with equality holds if and only if $\rho$ is real; (ii) since the square of the $l_2$-norm is a convex function of its argument, the corresponding measure of imaginarity is also convex; and  (iii) the monotonicity of $\mathcal{I}_{\mathbf{B}}(\rho)$ under the class of real unital operations. In order to prove (iii), let $\Lambda_R$ be a real unital channel, i.e. $\Lambda_R(I)=I$. {For this kind of } channel, $ \mathcal{I}_{\mathbf{B}}(\Lambda_R(\rho)) = \frac{1}{4}\|\Lambda_R(\rho) - (\Lambda_R(\rho))^T\|_2^2 = \frac{1}{4}\|\Lambda_R(\rho - \rho^T)\|_2^2 
\leq \frac{1}{4}\|\rho - \rho^T\|_2^2 = \mathcal{I}_{\mathbf{B}}(\rho)$,
where we have used the fact that for any physically consistent real operation, $(\Lambda_R(\rho))^T = \Lambda_R(\rho^T)$~\cite{Hickey2018}, and that for any traceless operator $X$, the inequality $\|\Lambda_R(X)\|_2^2 \leq \|X\|_2^2$ holds for unital channels~\cite{Prez-Garca2006,Wang2009}. Furthermore, note that $\mathcal{I}_{\mathbf{B}}(\rho)$ remains invariant under real orthogonal operations.

Despite this limitation, the $l_2$-norm offers a more tractable mathematical structure compared to the $l_1$-norm-based measure, as will be elaborated in the following section. Nevertheless, for pure states, the $l_2$-norm of imaginarity is monotone under all real operations and, even, satisfies strong monotonicity, i.e., it cannot increase on average under probabilistic real operations. Hence, it can serve as a natural and tractable quantifier of imaginarity in several physically relevant setting.
}

\subsection{Dynamical resource theory of imaginarity}
\label{sec:dynamical}

{Resource theories were originally formulated in a static framework, where quantum states constitute the primary objects, and the resource is regarded as an intrinsic property of these states, reflecting their usefulness for specific tasks under restricted operational capabilities. More recently, there has been significant interest in extending this framework to the dynamical regime~\cite{GiladGour2020,Gour2020prl,Ji_2025,Hsieh_2020,Saxena2020,Skrzypczyk2019,Liu2019,Liu2020}, where the goal is to quantify and compare the resourcefulness of quantum operations themselves. In such dynamical resource theories, quantum operations themselves serve as the fundamental objects, and their capability to generate or manipulate static resources determines their resourcefulness.}
{Specifically, the central objects in a DRT  are quantum channels $\Phi \in \Lambda^{d_1 \to d_2}$. {Transformations between channels are described by supperchannels, denoted by $\Pi^{(d_1,d_2)\rightarrow(d'_1,d'_2)}$, i.e., higher order maps that maps elements of $\Lambda^{d_1 \to d_2}$ to those of $\Lambda^{d_1' \to d_2'}$, i.e.,
\(
\Pi^{(d_1,d_2)\to(d'_1,d'_2)} = \{\Theta : \Lambda^{d_1 \to d_2} \xrightarrow{\Theta} \Lambda^{d_1' \to d_2'}\}\) while preserving the complete positivity and trace preservation ~\cite{Gour2019,Stratton2024}.} We now describe the general framework of a dynamical resource theory, consisting of three key components:
(i) {\it Free superoperations}($\mathcal{F_O^{\text{dyn}}} \subseteq \Pi^{(d_1,d_2)\to (d_1',d_2')} $) -- a class of higher-order maps that transform quantum channels in $\Lambda^{d_1 \to d_2}$ into quantum channels in $\Lambda^{d_1' \to d_2'}$; (ii) {\it Free operations} ($\mathcal{F_S^{\text{dyn}}} \subseteq \Lambda^{d_1 \to d_2}$) -- a set of certain channels considered to be resource free, that remains invariant under the action of free superoperations \(\Pi^{(d_1,d_2)\to (d_1',d_2')}\in\mathcal{F_O^{\text{dyn}}}\), i.e., \(\mathcal{F_S^{\text{dyn}}}\xrightarrow{\mathcal{F_O^{\text{dyn}}}}\mathcal{F_S^{\text{dyn}}}\); (iii) {\it Dynamical resource monotone} (\(f\)) -- a function that assigns a non-negative real value to each quantum channel, i.e., \(f:\Lambda^{d_1 \to d_2}\to\mathbb{R}_{\geq0}\), with the defining property that it does not increase under free superoperations, namely, i.e., $f(\Lambda^{d_1 \to d_2})\geq f(\Theta(\Lambda^{d_1 \to d_2}))~\forall~\Theta\in\mathcal{F}_{SO}$. In essence, it serves as a quantitative measure of the resource present in a given channel within the underlying resource-theoretic framework. 

Among the various notions of dynamical resources, one prominent approach is that of \textit{static resource preservability}, which examines how a given operation preserves the resource content of quantum states within an SRT~\cite{Hsieh_2020,Stratton2024}. In this setting, the free operations are those channels that are incapable of preserving the resource, commonly referred to as \textit{resource-breaking channels}~\cite{muhuri2025,Srinidhi2024,Bu2016,Vieira2025,Horodecki2003en,ayanpatra2024,Takagi2020}, and the corresponding free superoperations are those that leave this set invariant~\cite{Chen2020,Luo2022}. An alternative approach which we adopt here focuses on the ability of quantum operations to generate, on average, resource from the set of free states in the underlying SRT. In this case, the free operations are precisely those channels that cannot generate any resource from free states. Notably, this set coincides with the set of free operations in the SRT itself. The corresponding free superoperations are then defined as those that map the set of free operations of the underlying SRT onto itself. As in static resource theories, where pure states often receive particular attention due to their simple and more complete characterization~\cite{chitambar2019,Veitch2014,Leone2024,Vidal2000,Horodecki2009}, unitary dynamics in DRT typically allow a clearer and more tractable analysis than general quantum channels. Motivated by this, we focus here on the imaginarity-generating power of unitary operations within the framework of DRT.

To this end, we first outline the dynamical resource theory of imaginarity. In this framework, the free operations are given by the set of all real quantum channels, while the free superoperations are those that map the set of real quantum channels into itself, referred to as \textit{real superchannels}. Such real superchannels can be expressed as
\(
\Theta(\Phi) = \Xi_R^{\mathrm{post}} \circ (\Phi \otimes \mathcal{I}) \circ \Xi_R^{\mathrm{pre}},
\)
where $\mathcal{I}$ denotes the identity map, and $\Xi_R^{\mathrm{pre}}$ and $\Xi_R^{\mathrm{post}}$ represent real CPTP maps corresponding to pre- and post-processing operations, respectively. Restricting now to unitary dynamics, where the primary objects are unitary operators, we consider only those superoperations that map unitaries to either unitaries or ensembles of unitaries. If we fix $d_i = d_i' = d$ for $i=1,2$, both $\Xi_R^{\text{pre}}$ and $\Xi_R^{\text{post}}$ reduce to mixtures of orthogonal operations. Specifically, in such a scenario, one can write
\begin{equation}
    \Theta(U)[\cdot] = \sum_{i,j} q_i p_j \, O_i \, U \, \tilde{O}_j [\cdot] \tilde{O}_j^T \, U^\dagger \, O_i^T,
    \label{eq:free_superoperation}
\end{equation}
where $\{p_j\}$ and $\{q_i\}$ are valid probability distributions, and $\{O_i\}$ and $\{\tilde{O}_j\}$ are sets of $d$-dimensional orthogonal matrices.

In this setting, a valid monotone $f(U)$ must satisfy the following essential properties: $(i)$ ~\textit{positivity and faithfulness}: $f(U) \geq 0$, with equality if and only if $U$ is orthogonal (hence free); $(ii)$ ~\textit{monotonicity}: under a deterministic free superoperation $U \to \tilde{U}$, $f(U) \geq f(\tilde{U})$; and $(iii)$ ~\textit{strong monotonicity}: under a probabilistic free superoperation $U \to \{p_k, \tilde{U}_k\}$, $f(U) \geq \sum_k p_k f(\tilde{U}_k)$.

In the subsequent sections, we demonstrate that, within the regime of unitary dynamics, the average imaginarity-generating power of a unitary serves as a valid monotone in the dynamical resource theory of imaginarity, satisfying all the aforementioned properties. Furthermore, we explore the behavior of this quantity for typical unitaries, and show that, as the dimension increases, typical Haar random unitaries possess near--maximal amount of imaginarity-generating power--a feature that appears widely across various static resource theories.

\section{Imaginarity-generating power of Unitary dynamics: a resource theoretic perspective}
\label{sec:igpudvr}
Let us first introduce the notion of the imaginarity-generating power (IGP) of unitary dynamics and derive its closed-form expression. Furthermore, we show that the IGP constitutes a valid resource monotone for quantifying the resource content of unitary dynamics within the framework of the dynamical resource theory of imaginarity. To formulate this notion precisely, begin by defining the relevant sets of states. \\

\noindent\textbf{Definition 1.} (Real quantum states). \textit{The set of all real quantum states, written in a fixed basis $\{\ket{i}\}_{i=1}^d$, is defined as}
\begin{equation}
    \mathcal{F}_{R}=\{\rho_R:\braket{i}{\rho_R|j}\in\mathbb{R}~\forall~i,j\}.
\end{equation}
To incorporate a purity constraint, we introduce the following subset.\\

\noindent\textbf{Definition 1a.}(Real state with fixed purity). \textit{The set of real quantum states with fixed purity $\mathcal{P}$ is defined as}
\begin{equation}
\Sigma_{\mathcal{P}} = \left\{ \rho_R \in \mathcal{F}_{R} \;\middle|\; \Tr(\rho_R^{2}) = \mathcal{P} \right\}.
\label{eq:free_purity}
\end{equation}
Equipped with these notions, we are now ready to define the IGP of a unitary operator.\\

\noindent\textbf{Definition 2.} \textit{\textbf{purity-constrained imaginarity-generating power of a unitary.} For a unitary operator $U\in\mathcal{U}^d$, the purity-constrained imaginarity-generating power with respect to a basis $\mathbf{B}$, is defined as}
\begin{eqnarray}
   \nonumber \bar{\mathcal{I}}_{\mathbf{B}}(U)_{_\mathcal{P}} &=&\langle\mathcal{I}_{\mathbf{B}}(U\rho_{_R} U^\dagger)\rangle_{\rho_{_R}\in\Sigma_{\mathcal{P}}},\nonumber\\
&=&\int_{\Sigma_{\mathcal{P}}}  d\mu(\rho_{_R}) \; \mathcal{I}_{\mathbf{B}}(U \rho_{_R} U^\dagger),
\label{eq:1st_def}
\end{eqnarray}
\textit{where \(d\mu(\rho_{_R})\) represents the normalized Haar measure over the set \( \Sigma_{\mathcal{P}} \).}

This definition has a clear operational meaning and can be computed. To construct this measure explicitly, we begin with a diagonal density matrix \( \tilde{\rho}_{_R}(\vec\lambda) = \mathrm{diag}(\lambda_1, \lambda_2, \ldots, \lambda_d) \), whose eigen-spectrum $\vec{\lambda}=\{\lambda_1,...,\lambda_d\}$ satisfies the constraint \( \sum_{i=1}^d \lambda_i^2 = \mathcal{P} \). A generic real state with the same purity is then obtained via an orthogonal transformation, \( \rho_{_R} = O \tilde{\rho}_{_R} O^{\mathsf{T}} \), where \( O \) is a real orthogonal matrix. With this construction, the Haar average $d\mu(\rho_{_R})$ can be decomposed into two parts. First, one performs the Haar average over the orthogonal matrices for a fixed eigen-spectrum $\vec{\lambda}$ of states $\rho_{_R}(\vec\lambda) \in \Sigma_{\mathcal{P}}$. This is followed by an average over the set of eigen-spectra $\{\vec\lambda\}$ which satisfy the purity constraint, \(\sum_{i=1}^d \lambda_i^2 = \mathcal{P}\). Importantly, to ensure uniform sampling with respect to the Haar measure over $\Sigma_{\mathcal{P}}$, the second integration must be carried out with an appropriate measure over the distribution of $\vec{\lambda}$, which is generally nontrivial to determine. However, we will show later that the integration in Eq.~\eqref{eq:1st_def} does not depend on the distribution of $\vec\lambda$ (see the following subsection and Appendix~\ref{app:th_1} for a detailed discussion). Therefore, one can first perform the Haar average over the orthogonal group for a fixed spectrum, and subsequently average over all spectra satisfying the purity constraint $\mathcal{P}$, with respect to some probability distribution $d\nu(\vec{\lambda})$. Accordingly, we can write
\begin{eqnarray}
    \bar{\mathcal{I}}_{\mathbf{B}}(U)_{_\mathcal{P}} = \frac{1}{N} \int_{\vec\lambda} d\nu(\vec\lambda) \int_{O}d\mu(O) ~\mathcal{I}_{\mathbf{B}}(U O\tilde\rho_{_R}(\vec\lambda)O^TU^\dagger),
    \label{eq:correct_int}
\end{eqnarray}
where \(d\mu(O)\) denotes Haar measure over the orthogonal group, and \(N = \int_{\vec\lambda} d\nu(\vec\lambda)\), with $\vec\lambda$ satisfying the purity constraint $\sum_i\lambda_i^2=\mathcal{P}$, is the normalization constant.

\subsection{Imaginarity-generating power of qudit unitaries}
\label{sec:ipqu}

With the above framework in place, we evaluate the imaginarity-generating power of arbitrary unitary operators acting on $\mathbb{C}^d$. Let us first consider purity-constrained imaginarity-generating power of a unitary as defined in Definition $2$.
 \begin{theorem}
 For a fixed purity $\mathcal{P} = \mathrm{Tr}(\rho^2)$, the purity-constrained imaginarity-generating power of a unitary $U\in\mathcal{U}^d$, with respect to a basis $\mathbf{B}$, is given by
\begin{eqnarray}
       \bar{\mathcal{I}}_{\mathbf{B}}(U)_{_\mathcal{P}}&=&\frac{d^2-|\Tr(U^\dagger U^*)|^2}{2d(d+2)}\frac{d~\mathcal{P}-1}{d-1},
\label{eq:mixed_img}
\end{eqnarray}
where $U$ is represented in the basis $\mathbf{B}$.
     \label{th:qudit_mixed}
 \end{theorem}
 \begin{proof}
To evaluate $\bar{\mathcal{I}}_{\mathbf{B}}(U)_{_{\mathcal{P}}}$, as defined in Eq.~(\ref{eq:correct_int}), we first express the quantity $\mathcal{I}_{\mathbf{B}}(U \rho_{_R} U^\dagger)$ using Eq.~(\ref{eq:imaginarity_measure}) as
\begin{eqnarray}
\nonumber \mathcal{I}_{\mathbf{B}}(U \rho_{_R} U^\dagger)
&=& \frac{1}{2}\Big(\mathrm{Tr}(\rho_{_R}^2) - \mathrm{Tr}(U \rho_{_R} U^\dagger U^* \rho_{_R} U^T)\Big) \\
&=& \frac{1}{2}\Big(\mathrm{Tr}(\rho_{_R}^2) - \mathrm{Tr}(M^* \rho_{_R} M \rho_{_R})\Big),
\label{eq:img_unitary}
\end{eqnarray}
where $M = U^\dagger U^*$, and $U^*$ denotes the complex conjugate of $U$ in the basis $\mathbf{B}$. Let us fix a diagonal density matrix $\tilde{\rho}_{_R}(\vec{\lambda})$, where the spectrum $\vec{\lambda}$ satisfies the purity constraint $\sum_i \lambda_i^2 = \mathcal{P}$, and consider $\rho_{_R} = O \tilde{\rho}_{_R} O^T$. We now perform the Haar average of $\mathcal{I}_{\mathbf{B}}(U \rho_{_R} U^\dagger)$ over the orthogonal group, as defined in Eq.~\eqref{eq:correct_int}. This yields
\begin{eqnarray}
&&\int_{O} d\mu(O)~{\mathcal{I}}_{\mathbf{B}}(UO \tilde{\rho}_{_R} O^TU^\dagger)\nonumber\\
&&=\frac{1}{2}\int_{O} d\mu(O)~\big(\mathrm{Tr} (\tilde\rho_{_R}^2)-\mathrm{Tr}(M^*O\tilde\rho_{_R}O^T MO\tilde\rho_{_R}O^T )\big)\nonumber\\
&&=\frac{1}{2}\big(\big\langle \mathrm{Tr} (\tilde\rho_{_R}^2)\big\rangle-\big\langle \mathrm{Tr}(M^*O\tilde\rho_{_R}O^T MO\tilde\rho_{_R}O^T )\big\rangle_{O}\big)\nonumber\\
&&=\frac{d^2-|\mathrm{Tr}(U^\dagger U^*)|^2}{2d(d+2)}\frac{d\mathcal{P}-1}{d-1}.
\label{eq:im_u_p}
\end{eqnarray}
For a detailed calculation, see Appendix~\ref{app:th_1}. It is important to note that, Eq.~(\ref{eq:im_u_p}) depends only on the purity $\mathcal{P}=\sum_i \lambda_i^2$ and is independent of the individual eigenvalues. Consequently, it does not depend on the specific choice of the eigenvalue distribution $d\nu(\vec{\lambda})$. As a result, the integration over the eigenvalue distribution, along with the normalization constant $N$, does not affect the final expression in Eq.~\eqref{eq:im_u_p}. Therefore, we obtain
\begin{eqnarray}
\bar{\mathcal{I}}_{\mathbf{B}}(U)_{_\mathcal{P}} 
&=& \frac{d^2-|\mathrm{Tr}(U^\dagger U^*)|^2}{2d(d+2)}\frac{d\mathcal{P}-1}{d-1}.
\label{eq:only_haar}
\end{eqnarray}
This completes the proof.
 \end{proof}

{Furthermore, when restricting to the set of real pure states, i.e., $\Sigma_{\mathcal{P}}$ with $\mathcal{P}=1$, the expression in Eq.~(\ref{eq:mixed_img}) simplifies to a form that depends solely on the parameters of the unitary, leading to the following corollary.

\begin{corollary}
The imaginarity-generating power of a unitary $U\in\mathcal{U}^d$ with respect to a given basis $\mathbf{B}$, defined as the average imaginarity produced by $U$ acting on the set of real pure states, or equivalently, the purity-constrained imaginarity-generating power evaluated over $\Sigma_{\mathcal{P}}$ with $\mathcal{P}=1$, reads
\begin{eqnarray}
\bar{\mathcal{I}}_{\mathbf{B}}(U) = \bar{\mathcal{I}}_{\mathbf{B}}(U)_{_{\mathcal{P}=1}} = \frac{d^2 - |\mathrm{Tr}(U^\dagger U^*)|^2}{2d(d+2)}.
\label{eq:pure_img}
\end{eqnarray}
\label{co:pure_state}
\end{corollary}}

{\textbf{Remark~1.} Although fixing any value of $\mathcal{P}$ renders the quantity $\bar{\mathcal{I}}_{\mathbf{B}}(U)_\mathcal{P}$ independent of purity and dependent solely on the parameters of the unitary, the quantity $\bar{\mathcal{I}}_{\mathbf{B}}(U)$ exhibits an important distinction. In particular, the imaginarity measure in Eq.~\eqref{eq:imaginarity_measure} which we employed in Eq.~\eqref{eq:img_unitary}, is a valid monotone only under real unital channels. However, when restricted to pure states, Eq.~\eqref{eq:imaginarity_measure} becomes a valid monotone under the entire class of real channels. Consequently, while the imaginarity-generating power of a unitary $U \in \mathcal{U}^d$ evaluated over $\Sigma_{\mathcal{P}<1}$ implicitly assumes the set of free operations in both the SRT and DRT of imaginarity to be the class of real unital channels, restricting the analysis to $\Sigma_{\mathcal{P}=1}$ enlarges the set of free operations to include all real quantum channels. Therefore, the quantity $\bar{\mathcal{I}}_{\mathbf{B}}(U)$ serves as a suitable measure of dynamical imaginarity in a framework where the free operations comprise the full set of real channels.}

\subsubsection{Average of $\bar{\mathcal{I}}_{\mathbf{B}}(U)_{\mathcal{P}}$ over the purity $\mathcal{P}$}
\label{sec:averge_r}

Let us evaluate the IGP of a unitary operator $U$ over the set of all real states. In the preceding discussion, we derived the IGP, $\bar{\mathcal{I}}_{\mathbf{B}}(U)_{_{\mathcal{P}}}$, for a fixed purity $\mathcal{P}$, as given in Eq.~\eqref{eq:mixed_img}. We now extend this notion by averaging over all possible values of $\mathcal{P}$, thereby obtaining the IGP of $U$ over the entire set of real states. Since $\mathcal{P} = \frac{r^2 + 1}{d}$ with $r$ being the radius of the \(d\)-dimensional hypersphere, averaging over $\mathcal{P}$ is equivalent to averaging over the parameter $r^2$. We consider two choices for the distribution of $r$: (i) a uniform distribution and (ii) the Hilbert--Schmidt distribution.\\

\noindent\textbf{(i) Uniform distribution.}
We first analyze the case where $r$ is uniformly distributed over the interval $[0, \sqrt{d-1}]$, noting that $r$ reaches its maximum value for pure states in a $d$-dimensional Hilbert space. The average value of $r^2$ is then given by
\(\langle r^2 \rangle_{_{\text{UD}}} = \frac{(d-1)}{3}.\)
This yields the corresponding average purity \(
\langle \mathcal{P} \rangle_{_{\text{UD}}}= \frac{\langle r^2 \rangle_{_{\text{UD}}} + 1}{d}= \frac{2+d}{3d}.\)
Substituting into Eq.~\eqref{eq:mixed_img}, we obtain
\begin{eqnarray}
\langle \bar{\mathcal{I}}_{\mathbf{B}}(U)_{_{\mathcal{P}}} \rangle_{_{\text{UD}}}
&=& \frac{d^2 - |\mathrm{Tr}(U^\dagger U^*)|^2}{2d(d+2)}
\frac{d\langle \mathcal{P} \rangle_{_{\text{UD}}} - 1}{d-1} \nonumber\\
&=& \frac{\big(d^2 - |\mathrm{Tr}(U^\dagger U^*)|^2\big)}{6d(d+2)} \nonumber\\
&=&\frac{1}{3}\bar{\mathcal{I}}_{\mathbf{B}}(U),
\label{eq:avg_uni}
\end{eqnarray}
by using Eq.~\eqref{eq:pure_img}.

\noindent\textbf{(ii) Hilbert--Schmidt distribution.}
Next, we consider the Hilbert--Schmidt measure for the distribution of $r^2$. In this case, the average is given by \(\langle r^2 \rangle_{_{\mathrm{HSM}}} = \frac{d^2 - 1}{d^2 + 1}\)~\cite{Karol_Zyczkowski_2001}, which leads to the corresponding average purity
\[
\langle \mathcal{P} \rangle_{_{\mathrm{HSM}}}
= \frac{\langle r^2 \rangle_{_{\mathrm{HSM}}} + 1}{d}
= \frac{2d}{d^2 + 1}.
\]
Substituting this into Eq.~\eqref{eq:mixed_img}, we obtain
\begin{eqnarray}
\langle \bar{\mathcal{I}}_{\mathbf{B}}(U)_{_{\mathcal{P}}} \rangle_{_{\mathrm{HSM}}}
&=& \frac{d^2 - |\mathrm{Tr}(U^\dagger U^*)|^2}{2d(d+2)}
\frac{d \langle \mathcal{P} \rangle_{_{\mathrm{HSM}}} - 1}{d-1} \nonumber\\
&=& \frac{\big(d^2 - |\mathrm{Tr}(U^\dagger U^*)|^2\big)(d+1)}{2d(d+2)(d^2+1)}.
\label{eq:avg_HSM}
\end{eqnarray}
{These two expressions provide measures of imaginarity generating capability of a unitary where the input is drawn from different ensembles of real states. In both cases, the dependence of unitary remains entirely through $|\Tr(U^\dagger U^*)|^2$, while the choice of ensembles only modifies the overall prefactor. }

\subsection{Protocol for imaginarity-generating power detection}
\label{sec:pipd}

We now propose an experimentally feasible protocol to detect the imaginarity-generating power (IGP) of a unitary operation. From Eq.~\eqref{eq:mixed_img}, it follows that the IGP can be determined once the quantity $|\mathrm{Tr}(U^\dagger U^*)|^2$ is experimentally accessible.

Notably, this quantity can be rewritten as
\begin{eqnarray}
|\mathrm{Tr}(U^\dagger U^*)|^2 
&=& \mathrm{Tr}(U^T U)\,\mathrm{Tr}(U^\dagger U^*) \nonumber\\
&=& d^2 \left|\langle \phi^{+} | U \otimes U | \phi^{+} \rangle \right|^2,
\label{eq:ex_ob}
\end{eqnarray}
where $|\phi^{+}\rangle = \frac{1}{\sqrt{d}} \sum_{i=1}^{d} |ii\rangle$ denotes a maximally entangled state in $\mathbb{C}^d \otimes \mathbb{C}^d $. This relation enables an experimental scheme to estimate the IGP of a unitary $U$ via the following three-step protocol (see Fig.~\ref{fig:relation}):

\begin{enumerate}
    \item Prepare a maximally entangled state $|\phi^{+}\rangle = \frac{1}{\sqrt{d}} \sum_{i=1}^{d} |ii\rangle$.
    
    \item Prepare two copies of $U$ and apply them locally, i.e., implement $U \otimes U$ on the state $|\phi^{+}\rangle$, leading to the output state $\rho_{\mathrm{out}}=U \otimes U|\phi^{+}\rangle\langle\phi^{+}|U^\dagger \otimes U^\dagger$.
    
    \item Measure the fidelity $F = \langle \phi^{+} | \rho_{\mathrm{out}} | \phi^{+} \rangle$ of the resulting state with respect to $|\phi^{+}\rangle$ {from which one immediately obtains $|\Tr(U^\dagger U^*)|^2=d^2 ~F$}~\cite{bianchi2025,Zhang_2019,Qin_2024,Killoran2010,Ma2025}.
\end{enumerate}

\begin{figure}[ht]
\includegraphics[width=1.11\linewidth]{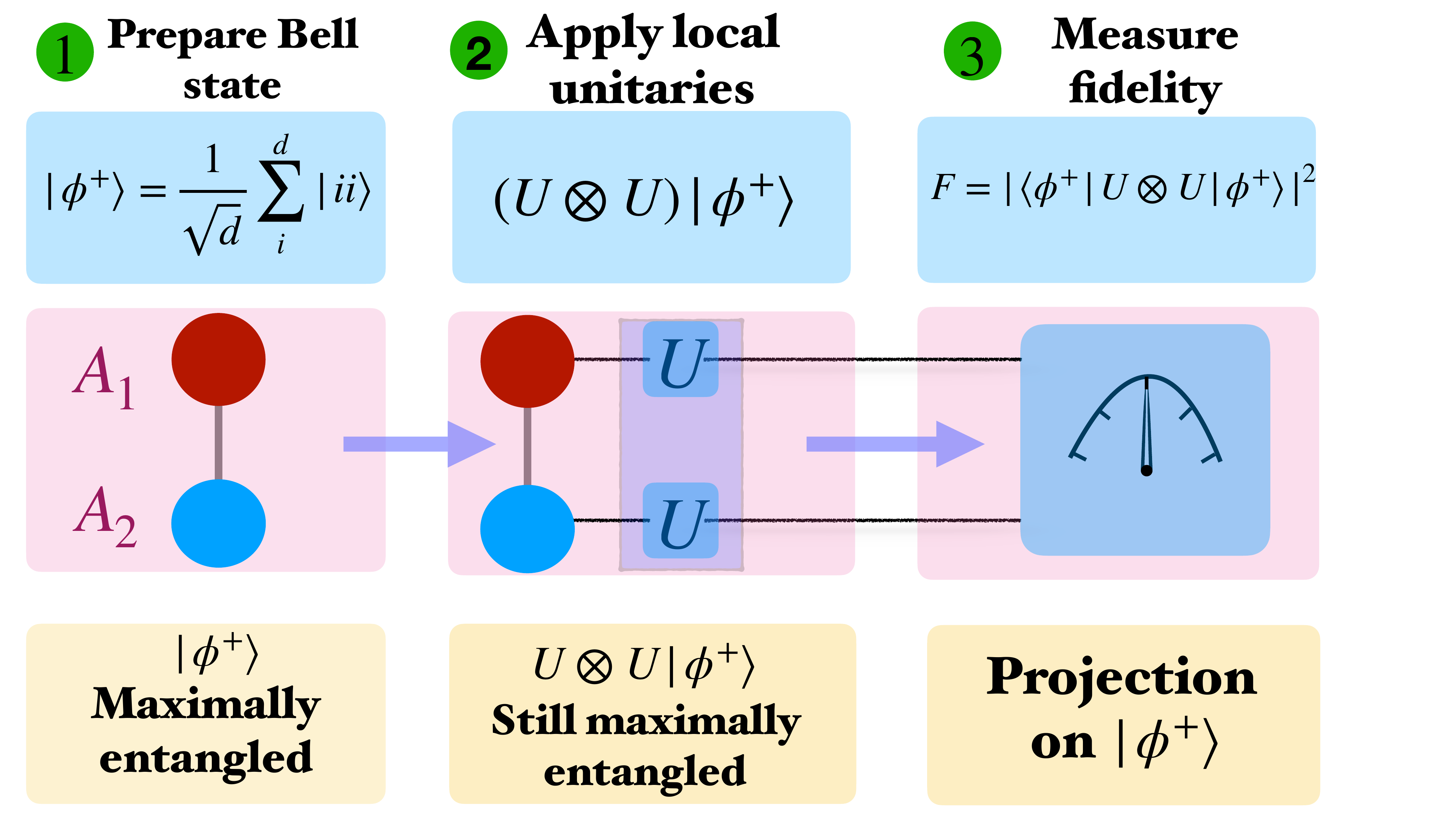}
\caption{
\textbf{Schematic illustration of the experimental protocol for quantifying the IGP.} {Alice $(A_1A_2)$ prepares a bipartite maximally entangled state $\ket{\phi^+}$ in dimension $d$, shared between subsystems $A_1$ and $A_2$. She then applies the unitary operation $U$ locally on both subsystems, implementing $U \otimes U \ket{\phi^+}$. Finally, she performs a projection onto the state $\ket{\phi^+}$. The corresponding projection probability is directly related to the IGP through Eqs.~\eqref{eq:ex_ob} and~\eqref{eq:mixed_img}.}
}
\label{fig:relation} 
\end{figure}

{Note that the fidelity estimation with maximally entangled states is regularly used in photonic, trapped-ion,superconducting plat forms \cite{Zhu2019,Lanyon_2009,Bentley2014,Somma2006,Malinowski2022,Ohfuchi_2024,Goss2022,Elder2020}. Further, a single measurement is enough to obtain the desired quantity. }

\subsection{Essential properties of imaginarity-generating power}
\label{subsec:resource_imag_power}

{Building on our discussion of the IGP under constrained purity, its average over purity, and the associated detection protocol, we now identify the properties required for the IGP to qualify as a valid resource monotone within the DRT framework. Ideally, these properties should hold for all values of $\mathcal{P}$. However, as noted earlier in Remark~$1$, we restrict our analysis to $\bar{\mathcal{I}}_{\mathbf{B}}(U)_{\mathcal{P}=1}$, since it incorporates the entire set of real operations as free operations in both the SRT and DRT of imaginarity. Inspired by its close analogy with entanglement and coherence generating power~\cite{Zanardi2000,Zanardi2017}, we now outline the key properties satisfied by the IGP of a unitary, $\bar{\mathcal{I}}_{\mathbf{B}}(U)$.}
 
{
\begin{property}
Given a basis $\mathbf{B} = \{\ket{i}\}_{i=1}^d$, the imaginarity-generating power of a unitary $U \in \mathcal{U}^d$ satisfies $\bar{\mathcal{I}}_{\mathbf{B}}(U) = 0$ if and only if $U$ does not generate imaginarity from any basis state $\ket{i}$.
\label{pro:imU_0}
\end{property}
\begin{proof}
Let us first assume that $\bar{\mathcal{I}}_{\mathbf{B}}(U) = 0$, which implies that $U$ does not generate imaginarity from any basis state in $\mathbf{B}$. Consequently, one must have \( U \ket{i} = \ket{\phi^{i}}_{R} \quad \forall \, i,\) where $\ket{\phi^{i}}_{R}$ denotes a state that is real in the reference basis $\mathbf{B}$.

Conversely, suppose that $U \ket{i} = \ket{\phi^{i}}_{R}$ for all $i$. Any real state $\ket{\psi}_{R}$ in the basis $\mathbf{B}$ can be written as \( \ket{\psi}_{R} = \sum_i c_i \ket{i},\) where all coefficients $c_i$ are real. Then, the action of $U$ on $\ket{\psi}_{R}$ is given by \( U \ket{\psi}_{R} = \sum_i c_i U \ket{i} = \sum_i c_i \ket{\phi^{i}}_{R}.\)
Since each $\ket{\phi^{i}}_{R}$ is a real state and the coefficients $c_i$ are real, it follows that $U \ket{\psi}_{R}$ is also real for any real input state $\ket{\psi}_{R}$. Therefore, $\bar{\mathcal{I}}_{\mathbf{B}}(U) = 0$, which completes the proof.
\label{proof:pro_4}
\end{proof}

\begin{property}
(\textbf{Positivity and faithfulness}) The imaginarity-generating power of $U$ satisfies $\bar{\mathcal{I}}_{\mathbf{B}}(U) \geq 0$ for all $U$, with equality if and only if $U$ is a real orthogonal operator.
\label{prop:faithfulpositive}
\end{property}

\begin{proof}
We show that $\bar{\mathcal{I}}_{\mathbf{B}}(U)$ satisfies both \textit{positivity} and \textit{faithfulness}.

\textit{\((1)\) Positivity.}
The condition $\bar{\mathcal{I}}_{\mathbf{B}}(U)\geq0$ follows from the bound  \(0 \leq \left|\Tr(U^\dagger U^*)\right|^2 \leq d^2\), since $U^\dagger U^*$ is itself a unitary operator.

\textit{\((2)\) Faithfulness.}
If $U$ is a real orthogonal operator in the reference basis $\mathbf{B} = \{\ket{i}\}_{i=1}^d$, it maps every real state to another real state, and hence cannot generate imaginarity. Therefore, $\bar{\mathcal{I}}_{\mathbf{B}}(U) = 0$. Conversely, suppose that $\bar{\mathcal{I}}_{\mathbf{B}}(U) = 0$. From Property~$1$, this implies that $U$ does not generate imaginarity from any basis state in $\mathbf{B}$. In particular, for each $\ket{j} \in \mathbf{B}$, we have
\begin{eqnarray}
U \ket{j} = \ket{\phi^j}_{R}, \quad \forall\, j,
\end{eqnarray}
where each $\ket{\phi^j}_{R}$ is real in the basis $\mathbf{B}$. Consequently, the matrix elements of $U$ can be written as
\begin{eqnarray}
U_{ij} = \langle i | U | j \rangle = \langle i | \phi^j \rangle_{R}.
\end{eqnarray}
Since both $\ket{i}$ and $\ket{\phi^j}_{R}$ are real, it follows that $U_{ij} \in \mathbb{R}$ for all $i,j$. Hence, $U$ is a real orthogonal matrix.
\end{proof}

\begin{property}
(\textbf{Invariance}) The quantity $\bar{\mathcal{I}}_{\mathbf{B}}(U)$ is invariant under deterministic superoperations of the form $U \to O_1 U O_2$, where $O_2$ and $O_1$ are real orthogonal matrices corresponding to pre- and post-processing operations, respectively. \label{cor:invariance_imaginarity_power_P}
\end{property}
\begin{proof}
Let $U \in \mathcal{U}^d$ be a unitary operator, and define another unitary $U_1 \in \mathcal{U}^d$ as $U_1 = O_1 U O_2$, where $O_1$ and $O_2$ are $d$-dimensional real orthogonal matrices. Using Eq.~\eqref{eq:img_unitary}, the quantity $\mathcal{I}_{\mathbf{B}}(U_1 \rho_{_R} U_1^\dagger)$ for $\rho_{_R} \in \Sigma_{\mathcal{P}=1}$ can then be written as
\begin{eqnarray}
\nonumber \mathcal{I}_{\mathbf{B}}(U_1 \rho_{_R} U_1^\dagger)
&=& \frac{1}{2}\Big(\mathrm{Tr}(\rho_{_R}^2) - \mathrm{Tr}(U_1 \rho_{_R} U_1^\dagger U_1^* \rho_{_R} U_1^T)\Big) \\
&=& \frac{1}{2}\Big(\mathrm{Tr}(\rho_{_R}^2) - \mathrm{Tr}(M_1^* \rho_{_R} M_1 \rho_{_R})\Big),
\label{eq:img_U1}
\end{eqnarray}
where \(M_1 = (O_1 U O_2)^\dagger (O_1 U O_2)^* = O_2^T U^\dagger U^* O_2 = O_2^T M O_2\),
with $M = U^\dagger U^*$, and $U^*$ denoting the complex conjugate of $U$ in the basis $\mathbf{B}$.

We now perform the Haar average over the set $\Sigma_1$, yielding
\begin{eqnarray}
&& \bar{\mathcal{I}}_{\mathbf{B}}(U_1) \nonumber\\
&&= \frac{1}{2}\Big(\big\langle \mathrm{Tr}(\rho_{_R}^2) \big\rangle_{\rho_{_R} \in \Sigma_1}
- \big\langle \mathrm{Tr}(M_1^* \rho_{_R} M_1 \rho_{_R}) \big\rangle_{\rho_{_R} \in \Sigma_1}\Big) \nonumber\\
&&= \frac{1}{2}\Big(\big\langle \mathrm{Tr}(\rho_{_R}'^{\,2}) \big\rangle_{\rho_{_R}' \in \Sigma_1}
- \big\langle \mathrm{Tr}(M^* \rho_{_R}' M \rho_{_R}') \big\rangle_{\rho_{_R}' \in \Sigma_1}\Big) \nonumber\\
&&= \bar{\mathcal{I}}_{\mathbf{B}}(U),
\label{eq:avg_img_U1}
\end{eqnarray}
where, in the second equality, we have used the substitution $\rho_{_R}' = O_2 \rho_{_R} O_2^T \in \Sigma_1$. This establishes the claim.
    \label{proof:col2}
\end{proof}}

{\noindent\textbf{Remark~2.} It is worth noting that the above properties can also be derived directly from the definition of the IGP given in Definition~2 (Eq.~\eqref{eq:1st_def}). In particular, the positivity of the IGP follows immediately from the positivity of $\mathcal{I}_{\mathbf{B}}(\rho)$. Likewise, faithfulness can be established by noting that, among all unitaries, real orthogonal matrices constitute the only free operations in the SRT of imaginarity. Furthermore, the invariance property follows from the fact that real orthogonal transformations preserve the set of real pure states, and that $\mathcal{I}_{\mathbf{B}}(\rho)$ remains invariant under such operations.}

{\begin{property}
(\textbf{Monotonicity}) The quantity $\bar{\mathcal{I}}_{\mathbf{B}}(U)$ remains invariant on average under the free superoperation $\Theta$, as defined in Eq.~\eqref{eq:free_superoperation} i.e., for the transformation $U \xrightarrow{\Theta} \{x_k, U_k\}$, it satisfies \( \bar{\mathcal{I}}_{\mathbf{B}}(U) = \sum_k x_k \, \bar{\mathcal{I}}_{\mathbf{B}}(U_k).\)
\end{property}
\begin{proof}
As we mentioned earlier in Eq.~\eqref{eq:free_superoperation}, the free deterministic superoperations which map a unitary to another unitary or an ensemble of unitaries can be written as 
\begin{eqnarray}
    \Theta(U)\left(\rho\right) &=& \sum_{i,j} q_i p_j \, O_i \, U \, \tilde{O}_j \rho \tilde{O}_j^T \, U^\dagger \, O_i^T\nonumber\\
    &=&\sum_{i,j} x_{ij} U_{ij}\rho U_{ij}^\dagger
    \label{eq:prop4}
\end{eqnarray}
where we define \(U_{ij} = O_i U\tilde{O}_j\) and $\{x_{ij}=q_i p_j\}$ forms a valid probability distributions. Here, $\{O_i\}$ and $\{\tilde{O}_j\}$ are sets of $d$-dimensional orthogonal matrices. Thus, under the action of $\Theta$, the unitary $U$ is mapped to an ensemble $\{x_{ij}, U_{ij}\}$. Invoking Property~\ref{cor:invariance_imaginarity_power_P}, it follows that under the transformation $U \to \{x_k, U_k\}$, one obtains
\begin{equation}
    \bar{\mathcal{I}}_{\mathbf{B}}(U) = \sum_k x_k \, \bar{\mathcal{I}}_{\mathbf{B}}(U_k).
    \label{eq:strong_monotonicity}
\end{equation}
This establishes that $\bar{\mathcal{I}}_{\mathbf{B}}(U)$ satisfies strong monotonicity under free superoperations, in fact with equality.
\end{proof}}

Imaginarity can arise in a quantum state at various stages of a quantum circuit or during the dynamical evolution of a system, both of which are governed by unitary operations. In the preceding discussion, we introduced the imaginarity-generating power of a unitary, which depends intrinsically on the properties of the unitary itself. In particular, Property~$2$ establishes that real orthogonal matrices possess vanishing IGP. We now turn to the characterization of unitaries that achieve maximal IGP. Using Corollary~\ref{co:pure_state}, we identify the class of unitaries that maximize the IGP. Consequently, within the dynamical framework of imaginarity resource theory, such unitaries correspond to those possessing the maximal amount of resource.
{
\begin{property}
\label{pro:maximum_IGP}
(\textbf{Maximum imaginarity-generating power}) The maximum value of $\bar{\mathcal{I}}_{\mathbf{B}}(U)$ is $\frac{d}{2(d+2)}$, and it is achieved by unitaries of the form $U = O_1 D O_2$, where $O_1$ and $O_2$ are real orthogonal matrices, and $D = \mathrm{diag}(e^{\iota \theta_1}, e^{\iota \theta_2}, \ldots, e^{\iota \theta_d})$, with the phases $\{\theta_i\}_{i=1}^d$ satisfying $\sum_{i=1}^d e^{2\iota \theta_i} = 0$. A notable special case corresponds to the generalized Pauli-$Z$ operator, given by $U_{ij} = \omega^{i m/2} \delta_{ij}$, where $\omega = e^{2\pi \iota / d}$ and $m \neq k d$ for $d \geq 2$, with $m,k \in \mathbb{Z}$.
\end{property}

\begin{proof}
Note that $\bar{\mathcal{I}}_{\mathbf{B}}(U)$ attains its maximum when $\left|\mathrm{Tr}(U^\dagger U^*)\right|^2$ is minimized, with the minimum value being zero. Observe that $|\Tr(U^\dagger U^*)|=|\Tr(U^TU)|$ where $U^T U$ is a symmetric unitary matrix and hence can be diagonalized by a real orthogonal matrix (see Appendix~\ref{App:diagM}). Therefore, one can write $U^TU=O_2^T D^2 O_2$, where $D = \mathrm{diag}(e^{\iota \theta_1}, e^{\iota \theta_2}, \ldots, e^{\iota \theta_d})$. Rewriting this as $U^TU=(O_1 D O_2)^T(O_1 D O_2)$, it follows that $U$ can be expressed in the form $U = O_1 D O_2$, where $O_1$ and $O_2$ are real orthogonal matrices. For such unitaries, one has \(\mathrm{Tr}(U^\dagger U^*) = \sum_{i=1}^d e^{2\iota \theta_i}\). Thus, in order to maximize $\bar{\mathcal{I}}_{\mathbf{B}}(U)$, the condition $\sum_{i=1}^d e^{2\iota \theta_i} = 0$ must be satisfied. 

As a concrete example, consider $U_{ij} = \omega^{i m/2} \delta_{ij}$, where $\omega = e^{2\pi \iota / d}$. In this case, \(\mathrm{Tr}(U^\dagger U^*) = \sum_{i=1}^d \omega^{i m}\), which vanishes for $m \not\equiv 0 \ (\mathrm{mod}\ d)$ with $m \in \mathbb{Z}$ and $d \geq 2$. This establishes the claim.
\end{proof}}

{With these properties in place, we are now in a position to introduce the imaginarity-generating power of a unitary, $\bar{\mathcal{I}}_{\mathbf{B}}(U)$, as a resource monotone for unitary dynamics within the framework of the DRT of imaginarity. A natural analogy arises from static resource theories, where pure-state monotones play a central role. For instance, the stabilizer R\'enyi entropy serves as a monotone for non-stabilizerness of pure states \cite{Leone2024,Leone2022}, while the entanglement entropy quantifies entanglement for pure states \cite{Vedral97,Vedral98,Vidal2000}. More generally, in many SRTs, monotones are often more tractable and well-understood for pure states, whereas their mixed-state counterparts are typically harder to characterize. In the dynamical setting, unitary dynamics may be viewed as the analogue of pure states, as they correspond to channels with pure Choi states. In contrast, general CPTP maps (non-unitary dynamics) are associated with environmental noise and correspond to mixed Choi states. With this perspective in mind, we now state one of the main results of this paper.

\begin{theorem} 
\textbf{IGP as a dynamical resource monotone.}
The imaginarity-generating power of a unitary constitutes a valid resource monotone for unitary dynamics in the framework of the DRT of imaginarity. In particular, it satisfies: (i) \textit{positivity}, (ii) \textit{faithfulness}, and (iii) \textit{monotonicity under free superoperations} $\Theta$, as defined in Eq.~\eqref{eq:free_superoperation}, such that $U \xrightarrow{\Theta} \{x_k, U_k\}$. 
Moreover, it attains its maximum value $\frac{d}{2(d+2)}$, and the corresponding maximally resourceful unitaries are of the form $U = O_1 D O_2$, where $O_1$ and $O_2$ are real orthogonal matrices, and $D = \mathrm{diag}(e^{\iota \theta_1}, e^{\iota \theta_2}, \ldots, e^{\iota \theta_d})$, with the phases $\{\theta_i\}_{i=1}^d$ satisfying $\sum_{i=1}^d e^{2\iota \theta_i} = 0$.
\end{theorem}

\begin{proof}
The proof follows directly from Properties~$1$--$5$ established above.
\end{proof}}

\section{imaginarity-generating power of typical unitaries}
\label{sec:typical_unitaries}

{In the previous section, we have analyzed the ability of a unitary operator to generate imaginarity on average from real states and identified the conditions under which the imaginarity-generating power becomes maximal. In particular, we have examined how the structural features of a unitary influence its capability to generate imaginarity. We now shift our focus to a complementary perspective, where instead of considering specific unitaries, we extend our analysis to the statistical behavior of the entire unitary group $\mathcal{U}^d$. This approach enables us to quantify the amount of imaginarity that can be {\it typically} generated on average by a randomly chosen unitary operator. Throughout this analysis, unitaries are sampled according to the Haar measure {on $\mathcal{U}^d$ with the goal of understanding the generic dynamical resourcefulness of quantum evolutions}.\\

\noindent\textbf{Definition 3.} \textit{\textbf{Purity-constrained average imaginarity-generating power of Haar-random unitaries.} The purity-constrained average imaginarity-generating power of unitaries drawn from the Haar-uniform distribution is defined as the Haar average of $\bar{\mathcal{I}}_{\mathbf{B}}(U)$ over the unitary group $\mathcal{U}^d$. Mathematically, it is given by
\begin{eqnarray}
\bar{\mathcal{I}}_{{\mathcal{P}}}
= \big\langle \bar{\mathcal{I}}_{\mathbf{B}}(U)_{_{\mathcal{P}}} \big\rangle_{U \in \mathcal{U}^d}
= \int_{\mathcal{U}^d} d\mu(U)\; \bar{\mathcal{I}}_{\mathbf{B}}(U)_{_{\mathcal{P}}}.
\label{eq:haar_uni}
\end{eqnarray}
Here, the integration is performed over the unitary group $\mathcal{U}^d$ with respect to the normalized Haar measure $d\mu(U)$.
}

\begin{proposition}
 The quantity $\bar{\mathcal{I}}_{\mathcal{P}}$ is basis independent.
     \label{pro:invariance}
\end{proposition}
\begin{proof}
The purity-constrained imaginarity-generating power of a unitary $U$ with respect to a basis $\mathbf{B}$ is given by
\begin{eqnarray}
    \bar{\mathcal{I}}_\mathbf{B}(U)_{\mathcal{P}}=\langle\mathcal{I}_{\mathbf{B}}(U\rho_{_R} U^\dagger)\rangle_{\rho_{_R}\in\Sigma_\mathcal{P}}.
\end{eqnarray}
Let the basis $\mathbf{B}$ be transformed to $\mathbf{B}'$ via a unitary $V$, i.e., $\mathbf{B}' = V \mathbf{B}$. Then, the matrix representation of an operator $X$ in the basis $\mathbf{B}$ is equivalent to $V^\dagger X V$ in the basis $\mathbf{B}'$. Consequently, we have
\begin{eqnarray}
    \bar{\mathcal{I}}_\mathbf{B}(U)_{\mathcal{P}}&=&\langle\mathcal{I}_{\mathbf{B}}(U\rho_{_R} U^\dagger)\rangle_{\rho_{_R}\in\Sigma_{\mathcal{P}}}\nonumber\\&=&\langle\mathcal{I}_{\mathbf{B}'}(V^\dagger U\rho_{_R} U^\dagger V)\rangle_{V^\dagger\rho_{_R}V\in\Sigma_{\mathcal{P}}}\nonumber\\&=&\bar{\mathcal{I}}_{\mathbf{B}'}(V^\dagger U V)_{\mathcal{P}}.
    \label{eq:basis_change_imag}
\end{eqnarray}
Using Eq.~\eqref{eq:basis_change_imag}, the purity-constrained average IGP over Haar-random unitaries, as defined in Definition~$3$, can be expressed as follows.
\begin{eqnarray}
\bar{\mathcal{I}}_{\mathbf{B},\mathcal{P}}&=&\int_{\mathcal{U}^d}d\mu(U)\; \bar{\mathcal{I}}_{\mathbf{B}}(U)_{_{\mathcal{P}}}\nonumber\\&=&\int_{\mathcal{U}^d}d\mu(U)\; \bar{\mathcal{I}}_{\mathbf{B}'}(V^\dagger U V)_{\mathcal{P}}\nonumber\\&=&\int_{\mathcal{U}^d}d\mu( U)\; \bar{\mathcal{I}}_{\mathbf{B}'}( U)_{\mathcal{P}}\nonumber\\&=&\bar{\mathcal{I}}_{\mathbf{B}',\mathcal{P}},
    \label{eq:basisinv}
\end{eqnarray}
where, in the third equality, we have used the invariance property of the Haar measure, i.e., for any integrable function $f$ and any $V,W \in \mathcal{U}^d$, one has $\int_{\mathcal{U}^d} f(U)\; d\mu(U) =\int_{\mathcal{U}^d} f(V U W)\;d\mu(U)$. Equation~\eqref{eq:basisinv} establishes the basis independence of $\bar{\mathcal{I}}_{\mathcal{P}}$.
\end{proof}
}
{Having established the basis independence of the purity-constrained average IGP of Haar-uniform unitaries, we now derive an explicit analytical expression for $\bar{\mathcal{I}}_{{\mathcal{P}}}$, as stated in the following theorem.
\begin{theorem}
Given a fixed purity \(\mathcal{P}\), the Haar-average of the imaginarity-generating power over the unitary group \(\mathcal{U}^d\) reduces to
\begin{eqnarray}
\bar{\mathcal{I}}_{\mathcal{P}}
&=&
\frac{d\mathcal{P}-1}{2(d+1)}.
\label{eq:img_u_free}
\end{eqnarray}
\label{pro:avg_u_fix_p}
\end{theorem}
\begin{proof}
Using the expression for $\bar{\mathcal{I}}_{\mathbf{B}}(U)_{{\mathcal{P}}}$ given in Eq.~\eqref{eq:mixed_img}, the Haar average over the unitary group $\mathcal{U}^d$ can be written as
\begin{eqnarray}   
\bar{\mathcal{I}}_{{\mathcal{P}}}&=&\langle \bar{\mathcal{I}}_{\mathbf{B}}(U)_{_\mathcal{P}}\rangle_{U\in{\mathcal{U}^d}}\nonumber\\&=&\frac{d^2-\big\langle|\Tr(U^\dagger U^*)|^2\big\rangle_{U\in{\mathcal{U}^d}}}{2d(d+2)}\frac{d~\mathcal{P}-1}{d-1},
\label{eq:avg_img_fix_purity}
\end{eqnarray}
where the quantity \(\big\langle |\Tr(U^\dagger U^*)|^2\big\rangle_{U\in{\mathcal{U}^d}}\) can be evaluated as
\begin{eqnarray}
&&\big\langle |\Tr(U^\dagger U^*)|^2\big\rangle_{U\in{\mathcal{U}^d}}\nonumber\\&=&\big\langle \Tr\big(U^T U\big) \Tr\big(U^{\dagger}U^*\big)\big\rangle_{U\in{\mathcal{U}^d}},\nonumber\\&=&\bigg\langle \sum_{i,j,m,n} U_{ji} U_{ji} U^{*}_{nm} U^{*}_{nm}\bigg\rangle_{U\in{\mathcal{U}^d}},\nonumber\\&=&\sum_{i,j,m,n}\big\langle  U_{ji} U_{ji} U^{*}_{nm} U^{*}_{nm}\big\rangle_{U\in{\mathcal{U}^d}},\nonumber\\&=&\sum_{i,j,m,n} \frac{2(d-1)}{d(d^2-1)}\delta_{jn}\delta_{im},\nonumber\\&=&\frac{2d}{d+1}.
    \label{eq:cal_uni_avg}
\end{eqnarray}
Here, we have employed the Weingarten calculus \cite{Collins2017}, which provides the Haar average of products of unitary matrix elements. Specifically, for a $d \times d$ unitary matrix $U$, with the integral of four matrix elements over the Haar measure $d\mu(U)$, one has
\begin{eqnarray}
&&\int_{U(d)} U_{i_1 j_1} U_{i_2 j_2} U^*_{i'_1 j'_1} U^*_{i'_2 j'_2} d\mu(U)\nonumber\\ = &&\sum_{\sigma, \tau \in S_2} \delta_{i_1, i'_{\sigma(1)}} \delta_{i_2, i'_{\sigma(2)}} \delta_{j_1, j'_{\tau(1)}} \delta_{j_2, j'_{\tau(2)}} \text{Wg}(\sigma^{-1} \tau, d),\nonumber
\end{eqnarray}
where the sum runs over permutations $\sigma, \tau \in S_2$, with $S_2$ denoting the symmetric group on two elements. The Weingarten function associated with the identity permutation $e = (1)(2)$ is given by $\mathrm{Wg}(e,d) = \frac{1}{d^2 - 1}$, while for the transposition $s = (12)$, it is $\mathrm{Wg}(s,d) = -\frac{1}{d(d^2 - 1)}$. Substituting Eq.~\eqref{eq:cal_uni_avg} into Eq.~\eqref{eq:avg_img_fix_purity} yields \(\bar{\mathcal{I}}_{{\mathcal{P}}} = \frac{d\mathcal{P}-1}{2(d+1)},\) which completes the proof.
\end{proof}

From Property~$5$ of the IGP of unitary operators, the maximum value of the purity-constrained IGP is given by \(\bar{\mathcal{I}}(U)_{_{\mathcal{P}}}^{\max} = \frac{d}{2(d+2)}\frac{d\mathcal{P}-1}{d-1}\). On the other hand, the average IGP of Haar-uniform unitaries associated to the same purity constraint is $\frac{d\mathcal{P}-1}{2(d+1)}$. In the large $d$ limit, both quantities converge to $\frac{\mathcal{P}}{2}$, and consequently their ratio approaches unity. This observation motivates the study of concentration properties of the IGP of Haar-random unitaries in high dimensions. To this aim, we introduce the normalized imaginarity-generating power as
\begin{eqnarray}
\bar{\mathcal{I}}_{\mathbf{B}}(U)_{N}= \frac{\bar{\mathcal{I}}_{\mathbf{B}}(U)_{_{\mathcal{P}}}}{\bar{\mathcal{I}}(U)_{_{\mathcal{P}}}^{\max}}.
\label{equ:norma_imga}
\end{eqnarray}
With this normalization, one obtains
\begin{eqnarray}
\bar{\mathcal{I}}_{\mathbf{B}}(U)_{N}=\frac{d^2 -  |\mathrm{Tr}(U^\dagger U^*)|^2 }{d^2}.
\label{eq:nor_im_pow}
\end{eqnarray}
We now employ L\'evy’s lemma to bound the probability that $\bar{\mathcal{I}}_{\mathbf{B}}(U)_{N}$ deviates from unity for Haar-random unitaries as the dimension increases. The resulting bounds are stated in the following proposition.

\begin{proposition}
{\textbf{(Typical near-maximality).}}
For a random unitary $U \in \mathcal{U}^d$ drawn according to the Haar measure, the normalized imaginarity-generating power $\bar{\mathcal{I}}_{\mathbf{B}}(U)_{_N}$ satisfies
\begin{eqnarray}
\Pr\!\left[ \bar{\mathcal{I}}_{\mathbf{B}}(U)_{_N} \ge 1 - \frac{3}{d^{2/3}} \right]
\;\ge\;
1 - \exp\!\left[-\frac{d^{2/3}}{64}\right],
\label{eq:Levy_lemma_result_for_imaginarity}
\end{eqnarray}
which implies that, in the large $d$ limit, deviations from the maximal imaginarity-generating power occur with exponentially small probability.
\label{pro:levy_lemma}
\end{proposition}

\begin{proof}
Let $X:\mathcal{U}^d \to \mathbb{R}$ be a real-valued function that is Lipschitz continuous with Lipschitz constant $K$ under the Hilbert--Schmidt norm, i.e., for all $U,V \in \mathcal{U}^d$, $|X(U) - X(V)| \le K \|U - V\|_2$, where $\|A\|_2 = \sqrt{\mathrm{Tr}(A^\dagger A)}$. L\'evy’s lemma then states that, if $U$ is sampled according to the Haar measure, for any $\varepsilon > 0$,
\begin{eqnarray}
\Pr\!\big[ X(U) - \langle X(U) \rangle_{\mathcal{U}^d} \ge \varepsilon \big]
\le
\exp\!\left[-\frac{d \varepsilon^2}{4K^2}\right].
\label{eq:levy_general1}
\end{eqnarray}
We now set $X(U) = 1 - \bar{\mathcal{I}}_{\mathbf{B}}(U)_{_N}$. Using Eqs.~\eqref{eq:cal_uni_avg} and \eqref{eq:nor_im_pow}, we obtain $\langle X(U)\rangle_{_{\mathcal{U}^d}}=\frac{2}{d(d+1)}$, which implies \(X(U)-\langle X(U)\rangle_{_{\mathcal{U}^d}}=1-\bar{\mathcal{I}}_{\mathbf{B}}(U)_{_N}-\frac{2}{d(d+1)}\). Substituting into Eq.~\eqref{eq:levy_general1}, we have \(\Pr\!\big[1-\bar{\mathcal{I}}_{\mathbf{B}}(U)_{_N}-\frac{2}{d(d+1)}\geq \varepsilon \big]\le\exp\!\left[-\frac{d\varepsilon^2}{4K^2}\right]\). In the large $d$ limit, this can be rewritten as \(\Pr\!\big[\bar{\mathcal{I}}_{\mathbf{B}}(U)_{N} \ge 1-\frac{2}{d^2}-\varepsilon \big]\ge 1-\exp\!\left[-\frac{d\varepsilon^2}{4K^2}\right]\). Setting $\varepsilon = d^{-\alpha}$ with $\alpha \in (0,1)$, we obtain
\begin{eqnarray}
\Pr\!\left[\bar{\mathcal{I}}_{\mathbf{B}}(U)_{_N} \ge 1 - \frac{3}{d^\alpha} \right]
\ge
1 - \exp\!\left[-\frac{d^{1-2\alpha}}{4K^2}\right].
\label{eq:Probability_inequality}
\end{eqnarray}
It remains to estimate the Lipschitz constant $K$. Since, \(X(U)=\frac{|\Tr(U^\dagger U^*)|^2}{d^2}=\frac{|\Tr(U^TU)|^2}{d^2}\), we can write
 \begin{eqnarray}
    &&|X(U)-X(V)|\nonumber\\&=& \frac{1}{d^2}\big||\Tr(U^TU)|^2-|\Tr(V^TV)|^2\big|\nonumber\\&\leq&\frac{1}{d^2}\big(\big|\Tr(U^TU-V^TV)\big|\big)\nonumber\big(\big| |\Tr(U^TU)|+|\Tr(V^TV)| \big|\big),\nonumber\\
&\leq&\frac{2}{d}~\bigg(\big|\Tr(U^T(U-V))\big|+\big|\Tr(V(U^T-V^T))\big|\bigg),\nonumber\\&\le&\frac{2}{d}\bigg(\|U^*||_{_2}\|(U-V)||_{_2}+\|V^*||_{_2}\|(U-V)||_{_2}\bigg),\nonumber\\&\le& \frac{4}{\sqrt{d}}~\|(U-V)\|_{_2}.
\label{eq:lipschitz1}
\end{eqnarray}
Here, we have used the bounds $|\mathrm{Tr}(AB)| \le \|A\|_2 \|B\|_2$, $|\mathrm{Tr}(U^T U)| \le d$, and $\|U^*\|_2 = \sqrt{d}$. Thus, the Lipschitz constant is $K = \frac{4}{\sqrt{d}}$. Finally, choosing $\alpha = \frac{2}{3}$ in Eq.~\eqref{eq:Probability_inequality}, we obtain
\begin{eqnarray}
\Pr\!\left[\bar{\mathcal{I}}_{\mathbf{B}}(U)_{_N} \ge 1 - \frac{3}{d^{2/3}} \right]
\ge
1 - \exp\!\left[-\frac{d^{2/3}}{64}\right],
\end{eqnarray}
which completes the proof.
\label{proof:levy_lemma1}
\end{proof}
Proposition~$2$ indicates that, as the dimension increases, Haar-random unitary operators tend to exhibit maximal resourcefulness in the DRT of imaginarity with probability approaching unity. To further obtain a complete characterization of the operational capabilities of such unitaries, in the following proposition, we also examine the fluctuations of the normalized imaginarity-generating power, which captures the spread of the normalized IGP around its mean value.
\begin{proposition}
In the large $d$ limit, the variance of the normalized imaginarity-generating power of a unitary, $\Delta^2 \bar{\mathcal{I}}_{\mathbf{B}}(U)_N$, scales as $O(d^{-4})$, where $d$ denotes the dimension of the unitary.
\label{pro:flu}
\end{proposition}

\begin{proof}
The normalized imaginarity-generating power $\bar{\mathcal{I}}_{\mathbf{B}}(U)_N$, as given in Eq.~\eqref{eq:nor_im_pow}, can be written as
\begin{eqnarray}
\bar{\mathcal{I}}_{\mathbf{B}}(U)_N = 1-\frac{y}{d^2},
\label{eq:img_xy}
\end{eqnarray}
where $y = |\Tr(U^\dagger U^*)|^2$. The variance of $\bar{\mathcal{I}}_{\mathbf{B}}(U)_N$ is then given by
\begin{eqnarray}
    \Delta^2 \bar{\mathcal{I}}_{\mathbf{B}}(U)_B &=& \langle (\bar{\mathcal{I}}_{\mathbf{B}}(U)_N)^2\rangle_{_{\mathcal{U}^d}} - \langle \bar{\mathcal{I}}_{\mathbf{B}}(U)_N\rangle^2_{_{\mathcal{U}^d}},\nonumber\\ 
    &=& \frac{1}{d^4}(\langle y^2\rangle_{_{\mathcal{U}^d}} - \langle y\rangle^2_{_{\mathcal{U}^d}}).
\end{eqnarray}
We now evaluate $\langle y^2 \rangle_{\mathcal{U}^d}$ as  
\begin{eqnarray}
    \langle y^2\rangle_{_{\mathcal{U}^d}} &=& \big\langle|\Tr(U^\dagger U^*)|^4\big\rangle_{\mathcal{U}^d},\nonumber\\
    &=& \big\langle\Tr\left((U^{T}U)\right)^2\Tr\left((U^*{^{T}U^{*}})\right)^2\big\rangle_{{\mathcal{U}^d}},\nonumber\\
    &=&\big\langle\sum_{\substack{i_1,i_2,i_1',i_2' \\ j_1,j_2,j_1',j_2'}}U_{i_1 j_1}^2U_{i_2 j_2}^2U^{*2}_{i_1^{\prime} j_1^{\prime}}U^{*2}_{i_2^{\prime} j_2^{\prime}}\big\rangle_{{\mathcal{U}^d}},\nonumber\\ &=&\sum_{\substack{i_1,i_2,i_1',i_2' \\ j_1,j_2,j_1',j_2'}}\big\langle U_{i_1 j_1}^2U_{i_2 j_2}^2U^{*2}_{i_1^{\prime} j_1^{\prime}}U^{*2}_{i_2^{\prime} j_2^{\prime}}\big\rangle_{{\mathcal{U}^d}}.
    \label{eq:ysq_avg}
\end{eqnarray}
Using the unitary Weingarten calculus~\cite{Collins2017}, one obtains
\begin{eqnarray}
    \sum_{\substack{i_1,i_2,i_1',i_2' \\ j_1,j_2,j_1',j_2'}}\big\langle U_{i_1 j_1}^2U_{i_2 j_2}^2U^{*2}_{i_1^{\prime} j_1^{\prime}}U^{*2}_{i_2^{\prime} j_2^{\prime}}\big\rangle_{{\mathcal{U}^d}}\approx 8+\frac{64}{d^{2}}+O(\frac{1}{d^{3}}),\nonumber\\
    \label{eq:Wingarten_final}
\end{eqnarray}
with the detailed derivation provided in Appendix~\ref{app:fluc}. Furthermore, from Eq.~\eqref{eq:cal_uni_avg}, one finds that $\langle y \rangle_{\mathcal{U}^d} = 2$ in the large $d$ limit. Substituting these results, we obtain
\begin{eqnarray}
      \Delta^2 \bar{\mathcal{I}}_{\mathbf{B}}(U)_N &\approx& \frac{1}{d^4}\bigg(4 +\frac{64}{d^{2}}+O(\frac{1}{d^3})\bigg),\nonumber\\ &\approx&\frac{4}{d^4}+\frac{64}{d^6}+O(\frac{1}{d^7}),\nonumber\\&\approx&\frac{4}{d^4}~~~~~~~~~~~~~~~~~\text{at}~d\to\infty.
\end{eqnarray}
Hence, the variance $\Delta^2 \bar{\mathcal{I}}_{\mathbf{B}}(U)_N$ scales as $O(d^{-4})$ in the large $d$ limit, completing the proof.
\end{proof}
{The above proposition illustrate that not only is the average IGP close to maximal, but the entire distribution collapses sharply around that value, since $d^{-4}$ suppression is very rapid.}

{\it Numerical verification of Proposition~3.} To further substantiate Proposition~\ref{pro:flu}, we  numerically  generate Haar-random unitaries across a range of dimensions and evaluate the variance $\Delta^2 \bar{\mathcal{I}}_{\mathbf{B}}(U)_N$. The numerical simulations show an excellent agreement with the theoretical scaling predicted as \(d^{-4}\) in Proposition~\ref{pro:flu}, especially for moderate to large $d$ regimes (e.g., $d \geq 20$), as depicted in Fig.~\ref{fig:relation}.}
\begin{figure}[ht]
\includegraphics[width=1.0\linewidth]{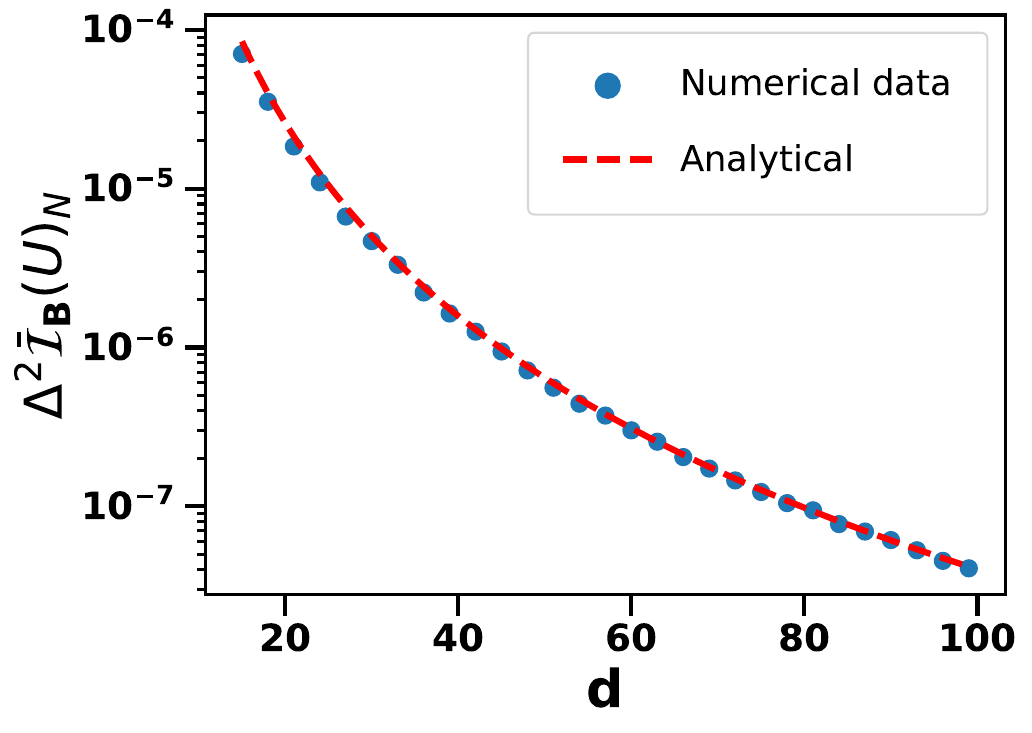}
\caption{ \textbf{Scaling of the variance of $\bar{\mathcal{I}}_{\mathbf{B}}(U)_N$ with dimension $d$.}
Log-scale plot of the numerically computed variance of $\bar{\mathcal{I}}_{\mathbf{B}}(U)_N$ as a function of the dimension $d$. The numerical data (blue points) are well described by the analytical scaling $\sim \frac{1}{d^4}(4+\frac{64}{d^2})$ (red dashed line). Both axes are dimensionless.}
\label{fig:relation} 
\end{figure}

\section{Conclusion}
\label{sec:con}


Complex numbers constitute an essential structural ingredient of quantum mechanics, with their significance recognized across 
diverse domains, including quantum communication \cite{Elliott2025}, quantum metrology \cite{Miyazaki2022}, and foundational aspects of quantum theory \cite{wu2021,Wu2024}. Beyond their mathematical necessity, the complex structure of quantum states has recently been identified as a valuable operational resource, formalized within the framework of the resource theory of imaginarity \cite{Hickey2018,wu2021}. Recent advances have further demonstrated that this resource can be experimentally detected in realistic quantum systems~\cite{chakrabarty2026}. While much of the existing work has focused on the characterization and quantification of imaginarity in quantum states, it is equally important to understand how this resource can be generated and manipulated through quantum dynamics. In particular, since quantum information processing tasks are primarily realized through unitary operations, a natural and fundamental question 
is how much imaginarity a given unitary operator can generate, on average, when acting on initially free, i.e., real quantum states.

Motivated by this question, we introduced the notion of imaginarity-generating power (IGP) of unitary operators 
within the framework of the dynamical resource theory (DRT) of imaginarity. In this setting, the set of free operations consists of all real operations, while the free superoperations are those that map the set of real operations onto itself. To quantify IGP of a unitary operator, we employed a measure of imaginarity based on the Hilbert--Schmidt norm, which we proved to be a valid monotone under real unital operations while also providing a mathematically tractable formulation. Notably, when restricted to pure states, this measure remains a valid monotone under all real operations.

Within this framework, we derived an exact analytical expression for the purity-constrained imaginarity-generating power of unitaries in arbitrary dimensions, and showed that, when restricting the imaginarity generation from pure real input states, the IGP reduces to a particularly simple form depending solely on the properties of the unitary. We further analyzed the average behavior of the IGP over ensembles of input states with varying purity, considering both uniform and Hilbert--Schmidt distributions. In addition, we proposed an experimentally feasible scheme to determine the IGP of a given unitary operator by measuring the fidelity of maximally entangled state after the action of two identical copies of the unitary, 
thereby establishing a direct connection between the theoretical framework and potential experimental implementation. Furthermore, we demonstrated that the IGP constitutes a valid resource monotone that quantifies the resourcefulness of unitary operators within the DRT of imaginarity, satisfying essential properties such as positivity, faithfulness, invariance, and monotonicity under free superoperations. In addition, we identified the class of unitary operators that maximize the IGP, along with the corresponding upper bounds. In terms of the statistical behavior of the IGP over Haar-random unitaries, we demonstrated that in large dimensions, the normalized IGP becomes strongly concentrated near its maximum value, with fluctuations fast decreasing as the dimension grows.

A natural open direction is to extend our results from unitary dynamics to general completely positive trace-preserving (CPTP) maps, which describe realistic noisy quantum processes. In this broader setting, several fundamental problems remain unresolved, including the formulation of a rigorous and operationally meaningful definition of imaginarity-generating power for general CPTP maps, the complete characterization of channels capable of generating imaginarity from real states, and the investigation of how this resource behaves under channel composition and in the presence of noise. The framework and results developed here will possibly stimulate deeper investigations into the role of imaginarity in quantum dynamics, information processing, and beyond.

\acknowledgements
 The authors acknowledge the use of cluster computing facility at the Harish-Chandra Research Institute.  A.K.A., M.S., S.M, A.P., and A.S.D. acknowledge the support from the project entitled `` Technology Vertical - Quantum Communication'' under the National Quantum Mission of the Department of Science and Technology (DST)  (Sanction Order No. DST/QTC/NQM/QComm/2024/2 (G)).

\onecolumngrid
\appendix

\section{Proof of Theorem 1}
\label{app:th_1}

Here, we compute the purity-constrained IGP of a given unitary for a fixed purity $\mathcal{P}$, where the free states consist of all elements in the set $\Sigma_{\mathcal{P}}$, represented in a chosen basis $\mathbf{B}$. According to Eq.~\eqref{eq:imaginarity_measure}, the imaginarity of a quantum state \(\rho\)  with respect to a fixed basis \(\textbf{B}\) is defined as
\begin{eqnarray}
\nonumber\mathcal{I}_{\mathbf{B}}(\rho) &=& \left\| \frac{\rho - \rho^*}{2} \right\|_{2}^{2}\\
&=& \frac{1}{4}\Tr(\rho^2+{\rho^*}^2-\rho\rho^*-\rho^*\rho)\nonumber\\&=&\frac12\Tr(\rho^2-\rho\rho^*).
\label{eq:img_state_def}
\end{eqnarray}
Using Eq.~\eqref{eq:img_state_def}, the imaginarity of the state, \(U\rho_{_R}U^\dagger\) can be express as
\begin{eqnarray}
\mathcal{I}_{\mathbf{B}}(U\rho_{_R}U^\dagger) = \frac{1}{2}\big(\Tr(\rho_{_R}^2)-\Tr(M^*\rho_{_R} M \rho_{_R} )\big),
\label{eq:img_un}
\end{eqnarray}
where \(M=U^\dagger U^*\), and we have used the fact that the purity of a state $\rho_{_R}$, i.e., $\Tr(\rho_{_R}^2)$, remains invariant under unitary transformations. To obtain the integrand in Eq.~\eqref{eq:correct_int}, we substitute \(\rho_{_R} = O \tilde{\rho}_{_R}(\vec{\lambda}) O^T\) into Eq.~\eqref{eq:img_un}, where $O$ is a real orthogonal matrix and the diagonal matrix $\tilde{\rho}_{_R}(\vec{\lambda})$ has elements $[\tilde{\rho}_{_R}(\vec\lambda)]_{ij} = \lambda_i \delta_{ij}$. This yields
\begin{eqnarray}
\mathcal{I}_{\mathbf{B}}\big(U O \tilde{\rho}_{_R}(\vec{\lambda}) O^T U^\dagger\big)
&=& \frac{1}{2} \bigg( \Tr\!\big(\tilde{\rho}_{_R}(\vec{\lambda})^2\big)
- \Tr\!\big(M^* O \tilde{\rho}_{_R}(\vec{\lambda}) O^T M O \tilde{\rho}_{_R}(\vec{\lambda}) O^T\big) \bigg).
\label{eq:app_urhou}
\end{eqnarray}
We now perform the Haar average of the expression in Eq.~\eqref{eq:app_urhou} over the orthogonal group. This leads to
\begin{eqnarray}
    \int_{O}d\mu(O) ~\bar{\mathcal{I}}_\mathbf{B}(UO\tilde{\rho}_{_R}(\vec{\lambda})O^TU^\dagger) &=& \frac{1}{2} \int_{O}d\mu(O)~\bigg(\Tr(\tilde{\rho}_{_R}(\vec{\lambda})^2) - \Tr(M^*O\tilde{\rho}_{_R}(\vec{\lambda})O^TMO\tilde{\rho}_{_R}(\vec{\lambda})O^T)\bigg)\nonumber\\ &=&\frac{1}{2}\Big(\mathcal{P}- \Big\langle\Tr(M^*O\tilde{\rho}_{_R}(\vec{\lambda})O^TMO\tilde{\rho}_{_R}(\vec{\lambda})O^T)\Big\rangle_{O}\Big),
    \label{eq:app_inte_2nd_or}
\end{eqnarray}
where the first term corresponds to the purity \(\mathcal{P}\) of the state \(\rho_{_R} \in \Sigma_{\mathcal{P}}\), and in the second term $d\mu(O)$ denotes the Haar measure on the $d$-dimensional orthogonal group. Furthermore, it is worth noting that \(M\) is a symmetric unitary matrix and hence can be diagonalized by a real orthogonal matrix, i.e., $M = \bar{O} D \bar{O}^T$ (see Appendix~\ref{App:diagM}). Therefore, the second term in Eq.~\eqref{eq:app_inte_2nd_or} can be rewritten as 
\begin{eqnarray}
 \nonumber \Big\langle\Tr(M^*O\tilde{\rho}_{_R}(\vec{\lambda})O^TMO\tilde{\rho}_{_R}(\vec{\lambda})O^T)\Big\rangle_{O}&=&{\Big\langle \Tr\Big(\bar{O}D^*\bar{O}^TO\tilde{\rho}_{_R}(\vec{\lambda}) O^T \bar{O}D\bar{O}^T O\tilde{\rho}_{_R}(\vec{\lambda}) O^T\Big)\Big\rangle}_{O}\nonumber\\
&=&{\Big\langle \Tr\Big(D^*\tilde{O}^T\tilde{\rho}_{_R}(\vec{\lambda}) \tilde{O}D \tilde{O}^T\tilde{\rho}_{_R}(\vec{\lambda})\tilde{O}\Big)\Big\rangle}_{\tilde{O}}\nonumber\\
&=&{\Big\langle \Tr\Big(D^*O^T\tilde{\rho}_{_R}(\vec{\lambda}) OD O^T\tilde{\rho}_{_R}(\vec{\lambda})O\Big)\Big\rangle}_{O},
\label{eq:Imaginarity_power_expression_in_terms_of_purity_and_orthogonal_matrix}
\end{eqnarray}
where, in the second line, we have defined $\tilde{O} = O^T \bar{O}$, and used the fact that averaging over $\tilde{O}$ is equivalent to averaging over $O$, allowing us to relabel $\tilde{O} \to O$ in the third line. Since $D$ is diagonal, it can be written as $D = \mathrm{diag}(\mu_1, \mu_2, \ldots, \mu_d)$, where $|\mu_j| = 1$. Using these expressions, the Eq.~(\ref{eq:Imaginarity_power_expression_in_terms_of_purity_and_orthogonal_matrix}) can be rewritten as
\begin{eqnarray}
 \nonumber\Big\langle \Tr\Big(D^*O^T\tilde{\rho}_{_R}(\vec{\lambda}) OD O^T\tilde{\rho}_{_R}(\vec{\lambda})O\Big)\Big\rangle_{O} &=& \nonumber \Big\langle \sum_{ijklmnpq}\mu_{i}^*\delta_{ij}O_{kj}\lambda_k\delta_{kl}O_{lm}\mu_m \delta_{mn}O_{pn}\lambda_p\delta_{pq}O_{qi} {\Big\rangle}_{O}\nonumber \\
&=& \nonumber\Big \langle \sum_{ikmp}\mu_{i}^*O_{ki}\lambda_kO_{km}\mu_m O_{pi}\lambda_pO_{pm} {\Big\rangle}_{O}\nonumber\\
&=& \nonumber \sum_{ikmp}\mu_{i}^*\lambda_k\mu_m\lambda_p \Big\langle O_{ki}O_{km} O_{pi}O_{pm} {\Big\rangle}_{O}.
\end{eqnarray}
The average over the orthogonal group can be evaluated using the orthogonal Weingarten calculus, yielding
\begin{eqnarray}
 \nonumber \Big\langle O_{ki}O_{km} O_{pi}O_{pm} {\Big\rangle}_{O}  &=&c_1(\delta_{im}+\delta_{kp}+\delta_{im}\delta_{kp})-c_2(1+\delta_{im}+\delta_{kp}+3\delta_{kp}\delta_{im}),\nonumber 
\end{eqnarray}
where $c_1 = \frac{d+1}{d(d-1)(d+2)}$ and $c_2 = \frac{1}{d(d-1)(d+2)}$~\cite{Collins2009}. Substituting this result back, we obtain
\begin{eqnarray}
 \nonumber \Big\langle\Tr\Big(OD^*O^T \tilde{\rho}_{_R}(\vec{\lambda}) OD O^T \tilde{\rho}_{_R}(\vec{\lambda})\Big)\Big\rangle_{O}  &=&  \sum_{ikmp}\mu_{i}^*\lambda_k\mu_m\lambda_p\big[c_1(\delta_{im}+\delta_{kp}+\delta_{im}\delta_{kp}) -c_2(1+\delta_{im}+\delta_{kp}+3\delta_{kp}\delta_{im})\big]\nonumber\\
 &=&c_1\bigg[\sum_i {|\mu_i|}^2+\sum_i \mu_i^*\sum_m \mu_m\sum_k\lambda_k^2+\sum_i {|\mu_i|}^2\sum_k \lambda_k^2\bigg]\nonumber \\&& - c_2\bigg[\sum_i \mu_i^*\sum_m \mu_m + \sum_i {|\mu_i|}^2+\sum_i \mu_i^*\sum_m \mu_m\sum_k\lambda_k^2 + 3\sum_i {|\mu_i|}^2\sum_k \lambda_k^2\bigg]\nonumber\\
 &=&c_1\bigg[d+\Tr(D)\Tr(D^*) \Tr(\tilde{\rho}_{_R}(\vec{\lambda})^{^2}) +d \Tr(\tilde{\rho}_{_R}(\vec{\lambda})^{^2})\bigg]\nonumber \\&& - c_2\bigg[\Tr(D)\Tr(D^*)+d+\Tr(D)\Tr(D^*) \Tr(\tilde{\rho}_{_R}(\vec{\lambda})^{^2}) + 3d \Tr(\tilde{\rho}_{_R}(\vec{\lambda})^{^2}) \bigg]\nonumber\\
 &=&c_1 \bigg[d (1+ \Tr(\tilde{\rho}_{_R}(\vec{\lambda})^{^2}) )+{|\Tr(D)|}^2 \Tr(\tilde{\rho}_{_R}(\vec{\lambda})^{^2}) \bigg]\nonumber\\&& - c_2\bigg[{|\Tr(D)|}^2 + d +{|\Tr(D)|}^2 \Tr(\tilde{\rho}_{_R}(\vec{\lambda})^{^2}) + 3d \Tr(\tilde{\rho}_{_R}(\vec{\lambda})^{^2}) \big]\nonumber\\
 &=&\frac{d^2+d(d-2)\mathcal{P} +{|\Tr(D)|}^2(d ~\mathcal{P}-1)}{d(d-1)(d+2)},
 \label{eq:app_int_orth_complete}
\end{eqnarray}
where we have used $\sum_k \lambda_k^2 = \Tr\!\big(\tilde{\rho}_{_R}(\vec{\lambda})^2\big) = \mathcal{P}$ and $\Tr(D) = \sum_i \mu_i$. Substituting this into Eq.~\eqref{eq:app_inte_2nd_or}, the Haar average over orthogonal matrices becomes
\begin{eqnarray}
    \int_{O}d\mu(O) ~\bar{\mathcal{I}}_\mathbf{B}(UO\tilde{\rho}_{_R}(\vec{\lambda})O^TU^\dagger) &=& \frac{d^2-{|\Tr(D)|}^2}{2d(d+2)}\frac{d \mathcal{P}-1}{d-1}\nonumber\\
    &=& \frac{d^2-{|\Tr(U^\dagger U^*)|}^2}{2d(d+2)}\frac{d ~\mathcal{P}-1}{d-1}.
\label{eq:Imaginarity_power_for_the_states_with_fixed_r}
\end{eqnarray}
From Eq.~\eqref{eq:Imaginarity_power_for_the_states_with_fixed_r}, it is evident that the Haar average over orthogonal matrices depends only on the purity $\mathcal{P}=\sum_k\lambda_k^2$ and is independent of the individual eigenvalues $\lambda_k$. Consequently, this quantity remains unchanged under integration over the spectral distribution in Eq.~\eqref{eq:correct_int}. Therefore, the imaginarity-generating power of a unitary $U$ can be written as
\begin{equation}
\bar{\mathcal{I}}_{\mathbf{B}}(U)_{_\mathcal{P}} = \frac{d^2-{|\Tr(U^\dagger U^*)|}^2}{2d(d+2)}\frac{d ~\mathcal{P}-1}{d-1}.
\end{equation} 
\section{Diagonalization of $M=U^T U$}
\label{App:diagM}
Notice that the matrix $M = U^T U$ is both unitary and symmetric, i.e., $M M^\dagger = M^\dagger M = I$ and $M = M^T$. We decompose $M$ into its real and imaginary parts as $M = A + iB$, where $A = \frac{M + {M}^{*}}{2}$ and $B = \frac{M - {M}^{*}}{2i}.$ Using the unitarity and symmetry conditions, namely $M^\dagger M = M^* M = I$, we obtain $A^2 + B^2 + i(BA - AB) = I.$ Equating the real and imaginary parts yields $A = A^T$, $B = B^T$, and $AB = BA$. Since $A$ and $B$ are real symmetric matrices that commute, they can be simultaneously diagonalized by a real orthogonal matrix $O$. Thus, one can write $A = O D_A O^T$ and $B = O D_B O^T$, which implies $M = A + iB = O(D_A + iD_B)O^T$. Therefore, $M$ admits the decomposition
\begin{eqnarray}
M = O D_M O^T,
\end{eqnarray}
where $D_M = D_A + iD_B$. Finally, from $M^* M = I$, it follows that $D_M^* D_M = I$, i.e., $|D_M|_{ij}^2 = \delta_{ij}$.

\section{Proof of Proposition 3}
\label{app:fluc}
Here, we evaluate the variance of $\bar{\mathcal{I}}_{\mathbf{B}}(U)_{_N}$ over the unitary group $\mathcal{U}^d$ with respect to the Haar measure. The variance is defined as
\begin{eqnarray}
 \Delta^2 \bar{\mathcal{I}}_{\mathbf{B}}(U)_{_N} &=& \langle (\bar{\mathcal{I}}_{\mathbf{B}}(U)_{_N})^2\rangle_{_{\mathcal{U}^d}} - \langle \bar{\mathcal{I}}_{\mathbf{B}}(U)_{N}\rangle^2_{_{\mathcal{U}^d}}.
 \label{app:varia_1}
\end{eqnarray}
Using the expression $\bar{\mathcal{I}}_{\mathbf{B}}(U)_{_N} = 1 - \frac{y}{d^2}$, where $y = |\Tr(U^\dagger U^*)|^2$, the variance can be written as
\begin{eqnarray}
\Delta^2 \bar{\mathcal{I}}_{\mathbf{B}}(U)_{_N}
= \frac{1}{d^4}(\langle y^2\rangle_{_{\mathcal{U}^d}} - \langle y\rangle^2_{_{\mathcal{U}^d}}).
\label{eq:xy_fluctuation}
\end{eqnarray}
Thus, the quantity $\Delta^2 \bar{\mathcal{I}}_{\mathbf{B}}(U)_{_N}$ depends only on the first and second moments of $y$. From Eq.~\eqref{eq:cal_uni_avg}, one has
$\langle y\rangle_{_{\mathcal{U}^d}}=\frac{2d}{d+1}$. The second moment of $y$ can be expressed as
\begin{eqnarray}
  \langle y^2\rangle_{_{\mathcal{U}^d}} &=& \big\langle|\Tr(U^\dagger U^*)|^4\big\rangle_{\mathcal{U}^d},\nonumber\\
    &=& \big\langle\Tr\left((U^{T}U)\right)^2\Tr\left((U^*{^{T}U^{*}})\right)^2\big\rangle_{{\mathcal{U}^d}},\nonumber\\
   &=&\sum_{\substack{i_1,i_2,i_1',i_2' \\ j_1,j_2,j_1',j_2'}}\big\langle U_{i_1 j_1}^2U_{i_2 j_2}^2U^{*2}_{i_1^{\prime} j_1^{\prime}}U^{*2}_{i_2^{\prime} j_2^{\prime}}\big\rangle_{{\mathcal{U}^d}},
   \label{eq:y2_avg_haar}
\end{eqnarray}
where the averaging is performed with respect to the Haar measure. To evaluate this quantity, we make use of the Weingarten calculus~\cite{Collins2017}, which provides the Haar average of products of unitary matrix elements. In particular, for a $d \times d$ unitary matrix $U$, the Haar integral involving eight matrix elements is given by
\begin{eqnarray}
&&\left\langle 
U_{i_1 j_1} \, U_{i_2 j_2} \, U_{i_3 j_3} \, U_{i_4 j_4} \,
{U}_{i'_1 j'_1}^* \, {U}_{i'_2 j'_2}^* \, {U}_{i'_3 j'_3}^* \, {U}_{i'_4 j'_4}^*
\right\rangle_{\mathcal{U}^d} \nonumber \\
&=& \sum_{\sigma, \tau \in S_4}
\delta_{i_1,\, i'_{\sigma(1)}} \,
\delta_{i_2,\, i'_{\sigma(2)}} \,
\delta_{i_3,\, i'_{\sigma(3)}} \,
\delta_{i_4,\, i'_{\sigma(4)}}
\;
\delta_{j_1,\, j'_{\tau(1)}} \,
\delta_{j_2,\, j'_{\tau(2)}} \,
\delta_{j_3,\, j'_{\tau(3)}} \,
\delta_{j_4,\, j'_{\tau(4)}}
\;
\mathrm{Wg}(\tau \sigma^{-1}, d),
\label{app:wig_for_uni}
\end{eqnarray}
where the summation is taken over all permutations $\sigma, \tau \in S_4$, with $S_4$ denoting the symmetric group of four elements. The permutations in $S_4$ can be grouped into five distinct conjugacy classes corresponding to the cycle types $[1,1,1,1]$, $[2,1,1]$, $[2,2]$, $[3,1]$, and $[4]$~\cite{Rotman1995,Collins2006,Serre1977}. The associated Weingarten functions $\mathrm{Wg}(\pi,d)$ depend only on the cycle structure of the permutation $\pi \in S_4$, and not on the specific permutation itself.
In the large $d$ limit, the corresponding Weingarten functions exhibit the asymptotic scaling
\begin{equation}
\mathrm{Wg}(\pi,d) \approx 
\begin{cases}
\frac{1}{d^4}, & [1,1,1,1], \\
\frac{1}{d^5}, & [2,1,1], \\
\frac{1}{d^6}, & [2,2], \\
\frac{1}{d^6}, & [3,1], \\
\frac{1}{d^7}, & [4].
\end{cases}
\end{equation}
Therefore, in the large $d$ regime, the leading contribution arises from the identity permutation associated with the cycle type $[1,1,1,1]$. Keeping only the leading order contribution, we take $\pi = e$. This implies $\tau \sigma^{-1} = e$, which in turn gives $\tau = \sigma$. Consequently, the summation reduces to a sum over all permutations $\sigma \in S_4$, resulting in a total of $4! = 24$ contributions. Accordingly, Eq.~\eqref{app:wig_for_uni} can be rewritten as
\begin{eqnarray}
 &&  \left\langle 
U_{i_1 j_1} \, U_{i_2 j_2} \, U_{i_3 j_3} \, U_{i_4 j_4} \,
U^*_{i'_1 j'_1} \, U^*_{i'_2 j'_2} \, U^*_{i'_3 j'_3} \, U^*_{i'_4 j'_4}
\right\rangle_{\mathcal{U}^d} \nonumber \\ &\approx& \sum_{\sigma}
\delta_{i_1,\, i{'}_{\sigma(1)}} \,
\delta_{i_2,\, i'_{\sigma(2)}} \,
\delta_{i_3,\, i'_{\sigma(3)}} \,
\delta_{i_4,\, i'_{\sigma(4)}}
\;
\delta_{j_1,\, j'_{\sigma(1)}} \,
\delta_{j_2,\, j'_{\sigma(2)}} \,
\delta_{j_3,\, j'_{\sigma(3)}} \,
\delta_{j_4,\, j'_{\sigma(4)}}
\;
\mathrm{Wg}(e, d).
\label{app:gh_ww}
\end{eqnarray}
Comparing Eqs.~\eqref{eq:y2_avg_haar} and \eqref{app:wig_for_uni}, the indices can be identified as
\begin{eqnarray}
(i_1,i_2,i_3,i_4) &\equiv& (i_1,i_1,i_2,i_2),\nonumber\\
(j_1,j_2,j_3,j_4) &\equiv& (j_1,j_1,j_2,j_2),\nonumber \\(i'_1,i'_2,i'_3,i'_4) &\equiv& (i'_1,i'_1,i'_2,i'_2),~\text{and} \nonumber \\(j'_1,j'_2,j'_3,j'_4) &\equiv& (j'_1,j'_1,j'_2,j'_2).
\label{app:equ_cy}
\end{eqnarray}
Substituting the conditions in Eq.~\eqref{app:equ_cy} into Eq.~\eqref{app:gh_ww}, we obtain
\begin{eqnarray}     
&&\big\langle U_{i_1 j_1}^2U_{i_2 j_2}^2U^{*2}_{i_1^{\prime} j_1^{\prime}}U^{*2}_{i_2^{\prime} j_2^{\prime}}\big\rangle_{{\mathcal{U}^d}}\nonumber
\\ &\approx& \sum_{\sigma}
\delta_{i_1,\, i'_{\sigma(1)}} \,
\delta_{i_1,\, i'_{\sigma(2)}} \,
\delta_{i_2,\, i'_{\sigma(3)}} \,
\delta_{i_2,\, i'_{\sigma(4)}}
\;
\delta_{j_1,\, j'_{\sigma(1)}} \,
\delta_{j_1,\, j'_{\sigma(2)}} \,
\delta_{j_2,\, j'_{\sigma(3)}} \,
\delta_{j_2,\, j'_{\sigma(4)}}
\;
\mathrm{Wg}(e, d),\nonumber\\
&\approx& \Big(
4\, \delta_{i_1 i'_1}\, \delta_{i_2 i'_2}\, \delta_{j_1 j'_1}\, \delta_{j_2 j'_2}
+ 4\, \delta_{i_1 i'_2}\, \delta_{i_2 i'_1}\, \delta_{j_1 j'_2}\, \delta_{j_2 j'_1}
+ 16\, \delta_{i_1 i'_1}\, \delta_{i_1 i'_2}\, \delta_{i_2 i'_1}\, \delta_{i_2 i'_2}\,
\delta_{j_1 j'_1}\, \delta_{j_1 j'_2}\, \delta_{j_2 j'_1}\, \delta_{j_2 j'_2}
\Big)\,\mathrm{Wg}(e,d).
\label{eq:yedhy}
\end{eqnarray}
In the above expression, the first term receives contributions from the permutations $(1)(2)(3)(4)$, $(12)(3)(4)$, $(1)(2)(34)$, and $(12)(34)$; the second term arises from $(13)(24)$, $(14)(23)$, $(1324)$, and $(1423)$; while the third term originates from the remaining $16$ elements of the group. Substituting the Haar-average expression obtained in Eq.~\eqref{eq:yedhy} into Eq.~\eqref{eq:y2_avg_haar}, we obtain
\begin{eqnarray}
    \langle y^2\rangle_{_{\mathcal{U}^d}} &\approx&\sum_{\substack{i_1,j_1,i_2,j_2 \\ i_1',j_1',i_2',j_2'}}  \Bigg(
4\, \delta_{i_1 i'_1}\, \delta_{i_2 i'_2}\, \delta_{j_1 j'_1}\, \delta_{j_2 j'_2}
+ 4\, \delta_{i_1 i'_2}\, \delta_{i_2 i'_1}\, \delta_{j_1 j'_2}\, \delta_{j_2 j'_1}
+ 16\, \delta_{i_1 i'_1}\, \delta_{i_1 i'_2}\, \delta_{i_2 i'_1}\, \delta_{i_2 i'_2}\,
\delta_{j_1 j'_1}\, \delta_{j_1 j'_2}\, \delta_{j_2 j'_1}\, \delta_{j_2 j'_2}
\Bigg)\,\mathrm{Wg}(e,d),\nonumber\\&\approx&  (4d^4+4d^4+16d^2)\bigg(\frac{d^4-8d^2+6}{d^2(d^2-1)(d^2-4)(d^2-9)}\bigg),
\label{eq:haar_y2_f}
\end{eqnarray}
where we have used \(\mathrm{Wg}(e,d)=\frac{d^4-8d^2+6}{d^2(d^2-1)(d^2-4)(d^2-9)}\). In the large $d$ limit, this yields the approximations
\begin{eqnarray}
    \langle y^2\rangle_{_{\mathcal{U}^d}} &\approx&8+\frac{64}{d^2}+{O}(\frac{1}{d^3}),~\text{and}\nonumber\\\langle y\rangle_{_{\mathcal{U}^d}} &=&\frac{2d}{d+1}\approx 2,
    \label{eq:all_approximation}
\end{eqnarray}
where the expression for $\langle y \rangle_{_{\mathcal{U}^d}}$ follows from Eq.~\eqref{eq:cal_uni_avg}. Substituting Eq.~\eqref{eq:all_approximation} into Eq.~\eqref{eq:xy_fluctuation}, we obtain
\begin{eqnarray}
\Delta^2 \bar{\mathcal{I}}_{\mathbf{B}}(U)_{_N}&=&\frac{1}{d^4}(\langle y^2\rangle_{_{\mathcal{U}^d}} -\langle y\rangle^2_{_{\mathcal{U}^d}}),\nonumber\\ &\approx& \frac{1}{d^4}\bigg(4 +\frac{64}{d^{2}}+O(\frac{1}{d^3})\bigg),\nonumber\\ &\approx&\frac{4}{d^4}+\frac{64}{d^6}+O(\frac{1}{d^7}),\nonumber\\&\approx&\frac{4}{d^4}~~~~~~~~~~~~~~~~~\text{at}~d\to\infty.  
      \label{eq:final_fluc_haar_avg}
\end{eqnarray}
Therefore, the variance $\Delta^2 \bar{\mathcal{I}}_{\mathbf{B}}(U)_{_N}$ scales as $\frac{1}{d^4}$ in the large $d$ limit, where $d$ denotes the dimension of the unitary operator. This completes the proof.
\twocolumngrid
\bibliography{bib}

\end{document}